\documentclass[10pt]{extarticle}
\pdfoutput=1
\usepackage[left=1in,top=1in,right=1in,bottom = 1in, nohead]{geometry}

\renewenvironment{thebibliography}[1]{%
	\begin{oldthebibliography}{#1}%
		\small
		\setlength{\parskip}{0.1ex}%		% 1.8 is normal
		\setlength{\itemsep}{0.0ex}%
	}%
	{%
	\end{oldthebibliography}%
}
\makeatletter
\def\@cite#1#2{[\textbf{#1\if@tempswa , #2\fi}]}
\def\@biblabel#1{[\textbf{#1}]}
\makeatother
\usepackage{mdframed}
\usepackage{graphicx}   % need for figures
\usepackage{amsmath,amssymb}    % need for subequations
\usepackage{thm-restate}
\usepackage{mdwlist}
\usepackage{mathtools}
\usepackage{xspace}
\usepackage{verbatim}   % useful for program listings
\usepackage{xcolor}
\usepackage{makecell}
\usepackage[mathscr]{euscript}
\usepackage{tikz}
\usepackage{relsize}
\usepackage{booktabs}
\usepackage{tikz-qtree}
\usepackage{subcaption}
\usepackage{multirow}
\usetikzlibrary{arrows.meta}
\usetikzlibrary{fit,shapes,trees,shapes.geometric}
\usetikzlibrary{patterns,decorations.pathreplacing,calc}
\usetikzlibrary{matrix, positioning, arrows}
\usetikzlibrary{chains,shapes.multipart}
\usetikzlibrary{shapes,calc}
\usetikzlibrary{automata}
\usepackage[linesnumbered,algoruled, lined, noend]{algorithm2e}

\usepackage[normalem]{ulem}

\newcommand{\xiao}[1]{{\color{blue} Xiao: [{#1}]}}

\usepackage{amsthm}

\newcommand{\cut}[1]{}   % remove things inside the cut
\makeatletter
\newcommand{\mytag}[2]{%
	\text{#1}%
	\@bsphack
	\protected@write\@auxout{}%
	{\string\newlabel{#2}{{#1}{\thepage}}}%
	\@esphack
}
\makeatother

\newenvironment{packed_enum}{
	\begin{enumerate}
		\setlength{\itemsep}{1pt}
		\setlength{\parskip}{0pt}
		\setlength{\parsep}{0pt}
	}
	{\end{enumerate}}

\newcommand{\introparagraph}[1]{\noindent {\bf \em #1.}}  % define own new subsection type: noindent, bold (textsc)

\usepackage{aliascnt} 

\newtheorem{definition}{Definition}

\newtheorem{lemma}{Lemma}

\newtheorem{theorem}{Theorem}

\usepackage{amsthm}
\usepackage{float}

\usepackage{tabularx,colortbl}
\usepackage[noend]{algpseudocode}
\usepackage{wrapfig}
\usepackage{lipsum}
\usepackage{arydshln}

\usepackage{graphicx,graphics}
\graphicspath{{figs/}}
\usepackage{xcolor}
\usepackage{subcaption}
\usepackage{thmtools, thm-restate}
\usepackage{mathrsfs}

\renewcommand{\paragraph}[1]{\medskip\noindent{\bf {#1. }}}
\renewcommand{\xiao}[1]{{\color{black} {#1}}}

\newcommand{\cost}{\mathsf{cost}}

\newcommand{\algo}{\mathbb{A}}
\newcommand{\dstr}{\mathcal{D}}
\newcommand{\copt}{C_{opt}}

\begin{document}
	
	\title{Algorithms for a Topology-aware Massively Parallel Computation Model}
	
	\author{
		Xiao Hu\\
		Duke University\\
		xh102@cs.duke.edu\\
	\and 
		 Paraschos Koutris\\
		UW-Madison\\
		paris@cs.wisc.edu\\
	\and
		Spyros Blanas\\
		The Ohio State University\\
		blanas.2@osu.edu\\
	}

	\date{}
	\maketitle
	
	\begin{abstract}
Most of the prior work in massively parallel data processing assumes homogeneity, i.e., every computing unit has the same computational capability, and can communicate with every other unit with  the same latency and bandwidth. However, this strong assumption of a uniform topology rarely holds in practical settings, where computing units are connected through complex networks. To address this issue, Blanas et al.~\cite{blanas2020topology} recently proposed a topology-aware massively parallel computation model that integrates the network structure and heterogeneity in the modeling cost. The network is modeled as a directed graph, where each edge is associated with a cost function that depends on the data transferred between the two endpoints. The computation proceeds in synchronous rounds, and the cost of each round is measured as the maximum cost over all the edges in the network.

In this work, we take the first step into investigating three fundamental data processing tasks in this topology-aware parallel model: set intersection, cartesian product, and sorting. We focus on network topologies that are tree topologies, and present both lower bounds, as well as (asymptotically) matching upper bounds. The optimality of our algorithms is with respect to the initial data distribution among the network nodes, instead of assuming worst-case distribution as in previous results. Apart from the theoretical optimality of our results, our protocols are simple, use a constant number of rounds, and we believe can be implemented in practical settings as well. 
\end{abstract}

%%Models for parallel computation typically ignore the underlying communication network, and thus cannot sufficiently capture certain intricacies of complex parallel systems. To address this issue, we propose a new theoretical model that 

%%Using the proposed model, we present lower and upper bounds for several fundamental data processing tasks for different network topologies. In particular, we consider set intersection, cartesian product and the binary join. All of our algorithms work in a single round, and achieve optimality with respect to the initial data distribution among the network nodes, instead of assuming worst-case distribution as in previous results. Notably, we present an almost optimal algorithm for set intersection in tree topologies, as well as an optimal algorithm for join computation on star networks.
	\section{Introduction}
\label{sec:intro}

The popularity of massively parallel data processing systems has led to an increased interest in studying the formal underpinnings of massively parallel models. As a simplification of the Bulk Synchronous Parallel (BSP) model~\cite{valiant1990bridging}, the Massively Parallel Computation (MPC) model~\cite{koutris2011parallel}, has enjoyed much success in studying algorithms for query evaluation~\cite{beame:communication, beame:skew, koutris2016worst, hu:similarity-joins, ketsman2017worst, hu2019output,tao2020binary}, as well as other fundamental data processing tasks~\cite{goodrich2011sorting, agarwal2016parallel, ghaffari2018improved,andoni2014parallel,barbosa2016new,GGKMR18,AssadiSW19}.  In the MPC model, any pair of compute nodes in a cluster communicates via a point-to-point channel. Computation proceeds in synchronous rounds: at each round, all nodes first exchange messages and then perform computation on their local data. 

Algorithms in the MPC model operate on a strong assumption of  homogeneity: every compute node has the same data processing capability and communicates with every other node with the same latency and bandwidth. In practice, however, large deployments are heterogeneous in their computing capabilities, often consisting of different generations of CPUs and GPUs. In the cloud, the speed of communication differs based on whether the compute nodes are located within the same rack, across racks, or across datacenters. In addition to static effects from the network topology, a model needs to capture the dynamic effects of different algorithms that may cause network contention. This homogeneity assumption is not confined in the theoretical development of algorithms, but it is also used when deploying algorithms in the real world.

Recent work has started taking into account the impact of network topology for data processing. In the model proposed by Chattopadhyay et al.~\cite{chattopadhyay2014topology,chattopadhyay2017tight}, the underlying network is modeled as a graph, where nodes communicate with their neighbors through the connected edges. Computation proceeds in rounds. In each round, $\tilde{O}(1)$\footnote{The notation $\tilde{O}$ hides a polylogarithmic factor on the input size.} bits can be exchanged per edge. The complexity of algorithms in such a model is measured by the number of rounds. Using the same model, Langberg et al.~\cite{langberg2019topology} prove tight topology-sensitive bounds on the round complexity for computing functional aggregate queries. Although these algorithmic results have appealing theoretical guarantees, they are unrealistic starting points for implementation. As the number of rounds required is usually polynomial in terms of the data size, the synchronization cost would be extremely high in practice. In addition, the size of the data that can be exchanged per edge in each round is too small; the compute nodes in today's mainstream parallel data processing systems can process gigabytes of data in each round.

Recently, Blanas et al.~\cite{blanas2020topology} proposed a new massively parallel data processing model that is aware of the network topology as well as network bandwidth. The underlying communication network is represented as a directed graph, where each edge is associated with a cost function that depends on the data transferred between the two endpoints. A subset of the nodes in the network consists of {\em compute nodes}, i.e., nodes that can store data and perform computation---the remaining nodes can only route data to the desired destination. Computation still proceeds in rounds: in each round, each compute node sends data to other compute nodes, receives data, and then performs local computation. There is no limit on the size of the data that can be transmitted per edge; the cost is defined as the sum across all rounds of the maximum cost over all edges in the network at each round. This model is general enough to capture the MPC model as a special case.

\renewcommand{\arraystretch}{1.3}
\begin{table*}[t]
	\centering
	\begin{tabular}{c|c|c|c}
		\hline
		\textbf{Task} & \textbf{Algorithm}  & \textbf{\# Rounds} & \textbf{Optimality Guarantee} \\
		\hline
		Set intersection & randomized & 1 & $O(\log |V| \log N)$ with high probability \\
		\cline{1-4}
		Cartesian product & deterministic &  1  & $O(1)$ \\
		\cline{1-4}
		{Sorting} & randomized  & {$O(1)$} & $O(1)$ with high probability \\
		\hline
	\end{tabular}
	\caption{A summary of our results. The graph network is $G = (V,E)$, while the size of the input data is denoted by $N$.}
	\label{tab:summary}
\end{table*}

In this work, we use the above topology-aware model to prove lower bounds and design algorithms for three fundamental data processing tasks: {\em set intersection}, {\em cartesian product}, and {\em sorting}. These three tasks are the essential building blocks for evaluating any complex analytical query in a data processing system.

In contrast to prior work, which either assumes a worst-case or uniform initial data distribution over the nodes in the network, we study algorithms in a more fine-grained manner by assuming that the cardinality of the initial data placed at each node can be arbitrary and is known in advance. This information allows us to build more optimized algorithms that can take advantage of data placement to discover a more efficient communication pattern.

\medskip \noindent \introparagraph{Our contributions}
We summarize our algorithmic results in Table~\ref{tab:summary}. Our results are restricted to network topologies that have two properties. First, they are {\em symmetric}, i.e., for each link $(u,v)$ there exists a link $(v,u)$ with the same bandwidth. Second, the network graph is a {\em tree}. Even with these two restrictions, we can capture several widely deployed topologies, such as star topologies and fat trees. 
All our algorithms are simple to describe and run either in a single round or in a constant number of rounds, hence requiring minimal synchronization. We thus believe that they form a good starting point for an efficient practical implementation.
We next present our results for each data processing task in more detail.

%%\begin{packed_grep}
\medskip \noindent {\bf Set Intersection} (Section~\ref{sec:set-intersection}). In this task, we want to compute the intersection $R \cap S$ of two sets.
Our lower bound for set intersection uses classic results from communication complexity on the lopsided set disjointness problem.
 This lower bound has a rather complicated form (as shown in Section~\ref{sec:set-intersection-lb}), since each link has a different data capacity budget depending on the underlying network as well as the initial data distribution. Since set intersection is a computation-light but communication-heavy task, the challenge is how to effectively route the data according to the capacity of each link. We design a single-round randomized routing strategy for set intersection that matches the lower bound with high probability, losing only a polylogarithmic factor (w.r.t. the input size and network size). Surprisingly, the routing depends only on the topology and initial data placement, but not the bandwidth of the links. 
 
 \medskip \noindent {\bf Cartesian Product} (Section~\ref{sec:cartesian-product}). Here we want to compute the cross product $R \times S$ of two sets. This task is fundamental for various join operators, such as natural join, $\theta$-join, similarity join and set containment join. We derive two lower bounds of different flavor. The first lower bound  has a similar form as that for set intersection. The second lower bound uses instead a counting argument, which states that each pair in the cartesian product must be enumerated by at least one compute node, and the two elements participating in this result should reside on the same node when it is enumerated. We propose a one-round deterministic routing strategy for computing the cartesian product, which has asymptotically optimal guarantees. Our protocol generalizes the HyperCube algorithm that is used to compute the cartesian product in the MPC model~\cite{afrati2011optimizing}.

  \medskip \noindent {\bf Sorting} (Section~\ref{sec:sorting}). We first define a valid ordering of compute nodes as any left-to-right traversal of the underlying network tree, after picking an arbitrary node as the tree root. If the ordering of compute nodes is $v_1, v_2,\cdots,v_{|V_C|}$, at the end of the algorithm all elements on node $v_i$ are in sorted order and smaller than those on node $v_j$ if $i < j$. Our lower bound again has a similar form to the one we derived for set intersection. We present a sampling-based sorting algorithm which runs in a constant number of rounds and matches our lower bound with high probability. The protocol is again independent of the topology and the bandwidth, and depends only on the initial placement of the data.
%%\end{packed_grep}

	\section{The Computational Model}
\label{sec:model}

In this section, we present the computational model we will use for this work.

\medskip \noindent \introparagraph{Network Model}
We model the network topology using a {\em directed} graph $G = (V,E)$. 
Each edge $e \in E$ represents a network link with bandwidth $w_e \geq 0$, where the direction of the edge captures the direction of the data flow.
We distinguish a subset of nodes in the network,  $V_C \subseteq V$, to be {\em compute} nodes.
Compute nodes are the only nodes in the network that can store data and perform computation on their local data. Non-compute nodes can only route data. We only consider connected networks, where every pair of compute nodes is connected through a directed path.

\medskip \noindent \introparagraph{Computation}
A parallel algorithm $\mathbb{A}$ proceeds in sequential {\em rounds} (or {\em phases}). 
We denote by $r \in \mathbb{N}$ the {number of rounds} of the algorithm.
In the beginning, each compute node $v \in V_C$ holds part of the input $I$, denoted $X_0(v) \subseteq I$. In this work, we assume that $\{X_0(v)\}_{v \in V_C}$ forms a partition of the input $I$; in other words, there is no initial data duplication across the nodes. The goal of the algorithm is to compute a function over the input $I$, such that in the end the compute nodes together hold the function output.

We also assume that the algorithm $\mathbb{A}$ has knowledge of the following: $(i)$ the topology of the graph, $(ii)$ the bandwidth of each link, and $(iii)$ $|X_0(v)|$ for each compute node $v \in V_C$. In the case of relational data, we further assume that the algorithm knows the cardinality of the local fragment for each relation.

We use $X_i(v)$ to denote the data stored at compute node $v \in V_C$ after the $i$-th round completes, where $i=1, \dots, r$. At every round, the compute nodes first perform some computation on their local data. Then, they communicate by sending data to other compute nodes in the network. 
We assume that for a data transfer from compute node $u$ to compute node $v$, the algorithm must 
explicitly specify the routing path (or a collection of routing paths). 
We use $Y_i(e)$ to denote the data that is routed through link $e$ during round $i$, and $|Y_i(e)|$ denote its total size measured in {\em bits}.

\medskip \noindent \introparagraph{Cost Model}
Since the algorithm proceeds in sequential rounds, we can decompose the cost of the algorithm, denoted $\cost(\algo)$, 
as the sum of the costs for each round $i$, 
\begin{align*} 
 \cost(\algo) = \sum_{i=1}^r \cost_i(\algo)
\end{align*}

The model captures the cost of each round by considering only the {\em cost of communication}. The cost of the $i$-th round is
\begin{align*} \label{eq:cost:model}
 \cost_i(\algo) := \max_{e \in E(G)} |Y_i(e)|/w_e.
 \end{align*}
In other words, the cost of each round is captured by the cost of transferring data through the most bottlenecked link in the network. 
 In some cases, it will be convenient to express the cost using tuples/elements instead of bits, which we will mention explicitly.

Even though the model does not take into account any computation time in the cost,  it is possible to incorporate computation costs in the model by appropriately transforming the underlying graph -- for more details, see~\cite{blanas14}. We should note here that our model does not capture factors such as congestion on a router node, or communication delays due to large network diameter.

\subsection{Network Topologies}

Even though the model supports general network topologies, computer networks often have a specific structure. When the underlying topology has some structure, several problems (such as routing~\cite{DBLP:journals/talg/BansalFKS14, DBLP:journals/tc/Leiserson85, DBLP:journals/siamcomp/ChekuriKS09, DBLP:conf/icalp/ChekuriEV12}) admit more efficient solutions than what is achievable for general topologies.
It is therefore natural to consider restrictions on the topology that are of either theoretical or practical interest. 

\medskip \noindent \introparagraph{Symmetric Network}
Wired networks support full duplex operation that allows simultaneous communication in both directions of a link.
Furthermore, datacenter networks allocate the same bandwidth for transmitting and receiving data for each node.
These networks are represented in the model using a symmetric network.
We say that a network topology is {\em symmetric} if for every edge $e = (u,v) \in E$, we also have that $e' = (v,u) \in E$ with
$w_e = w_{e'}$. In other words, the cost of sending data from $u$ to $v$ is the same as the cost of sending the same data from $v$ to $u$.

\begin{figure}
	\centering
	\begin{subfigure}[b]{0.4\columnwidth}
		\centering
		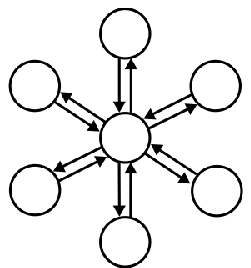
		\caption{Star topology.}
		\label{fig:topology:star}
	\end{subfigure}
	\hfill
	\begin{subfigure}[b]{0.5\columnwidth}
		\centering
		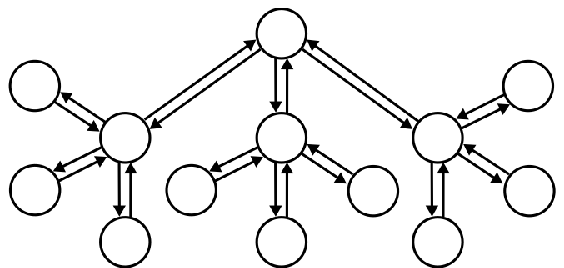
		\caption{Tree topology.}
		\label{fig:topology:tree}
	\end{subfigure}
	\caption{Common computer network topologies have structure, which permits more efficient solutions than what is feasible for arbitrary topologies.}
	\label{fig:topology}
\end{figure}

\medskip \noindent \introparagraph{Star Topology} 
The most common topology for small clusters is the star topology, where all computers are connected to a single switch.
A star network with $p+1$ nodes has $p$ compute nodes $V_C=\{v_1, \dots, v_p\}$ that are all connected to a central node $w$ that only does routing.
Figure~\ref{fig:topology:star} depicts an example of a star network. 
Within a node, a multi-core CPU also exhibits a star topology: individual CPU cores exchange data through a shared cache and memory hierarchy, which implicitly forms the center of the star. 

\medskip \noindent \introparagraph{Tree Topology}
As the network grows, a single router is no longer sufficient to connect all nodes.
A common solution to scale the network further is to arrange $r$ routers $\{w_1, \cdots, w_r\}$ in a star topology, and connect $p$ compute nodes $V_C=\{v_1, \cdots, v_p\}$ to individual routers.
Figure~\ref{fig:topology:tree} shows an example of a tree topology.
A key property in a tree topology is that there exists a unique directed path between any two compute nodes, hence routing is trivial.

%A tree topology with theoretical and practical interest is the \emph{fat tree} topology, where the capacity of the links increases as one gets closer to the root of the tree. 
%For example, the tree topology shown in Figure~\ref{fig:topology:tree} is a fat tree if $f_{(v_1, w_1)}(x) \gg f_{(w_1, w_4)}(x)$.
%Fat trees are theoretically interesting because it is known that in many cases routing is more efficient than in general trees~\cite{DBLP:journals/tc/Leiserson85}.
%Fat trees are practically interesting as they correspond to the common network practice of organizing large clusters in racks, where the rack-level switch uses multiple `up' links to the core of the network and a single `down' link to each node.

%\smallskip \noindent {\em Other Topologies.} \paris{Here, perhaps we mentioned some other topologies of interest, e.g. ring, grid. }

\smallskip

In this work, {\em we will focus on symmetric tree topologies}. We make two observations about such topologies:
\begin{itemize}
\item We can assume w.l.o.g. that every compute node is a leaf. Indeed, if we have a non-leaf compute node $v \in V_C$, we can transform $G$ to a new graph $G'$ by adding a new compute node $v'$, introduce a new link between $v,v'$ with bandwidth $+\infty$, and make $v$ a non-routing node.
\item We can assume w.l.o.g. that there are no nodes with degree 2. Indeed, consider a non-leaf node $v$ with two adjacent edges $e_1 = (v,u_1)$, $e_2=(v,u_2)$. We can then remove $v$, and replace the two edges with a single edge $e = (u_1, u_2)$ with bandwidth $\min\{w_{e_1}, w_{e_2}\}$. 	
\end{itemize}

\subsection{Relation to the MPC Model}

We discuss here how the topology-aware model can capture the MPC model~\cite{guide:mpc,beame:communication} as a special case. Recall that in the MPC model we have a collection of $p$ nodes. The MPC model is topology-agnostic: every machine  can communicate with any other machine, and the cost of a round is defined as the maximum amount of data that is received during this round across all machines. The MPC model corresponds to an asymmetric star topology with $p$ compute nodes. For every edge $e = (v_i, o)$ that goes from a compute node to the center $o$ the bandwidth is $w_e = + \infty$, while for the inverse edge $e' = (o,v_i)$ the bandwidth is $w_{e'} = 1$. 

It should be noted that all previous works using the MPC model assume a uniform data distribution, where each node initially receives $N/p$ data, where $N$ is the input size. This assumption has been used both for lower and upper bounds. In contrast, our algorithms and lower bounds take the sizes of the initial data distribution as parameters.

	\section{Set Intersection}
\label{sec:set-intersection}

%
%We consider set intersection under the edge-capacitated model, where the cost function is $f_e(x) = x/w_e$, where $w_e$ is the bandwidth of the link. 

In the set intersection problem,  we are given two sets $R, S$. Our goal is to enumerate all pairs $(r,s) \in R \cap S$. Note that there is no designated node for each output pair, as long as it is emitted by at least one node. We assume that all elements from both sets are drawn from the same domain.

Given an initial distribution $\dstr$ of the data across the compute nodes, we denote
by $R^\dstr_v$, $S^\dstr_v$ the elements from $R$ and $S$ respectively in node $v$.
Let $N^\dstr_v = |R^\dstr_v| + |S^\dstr_v|$, and $N^\dstr = \sum_v N_v = |R| + |S|$. 
Whenever the context is clear, we will drop the superscript $\dstr$ from the notation.

%We assume that all elements are drawn from the same domain, and that initially the input data is partitioned across the compute nodes. 

%%%%%%%%%%%%%%%%%%%%%%
\subsection{Lower Bound for Tree Topologies}
\label{sec:set-intersection-lb}

We present a lower bound on the cost for the case of a symmetric tree topology.
% in other words, if $e = (u,v)$ and $e' = (v,u)$, then $w_e = w_{e'}$. 
To prove the lower bound, we use a reduction from the {\em lopsided set disjointness} problem in communication complexity. 
In this problem, Alice holds a set $X$ of $n$ elements and Bob holds a set $Y$ of $m$ elements from some common domain. The goal is to decide whether the intersection $X \cap Y$ is empty by minimizing communication. It is known~\cite{Patrascu11,LSD12} that for any multi-round randomized communication protocol, either Alice has to send $\Omega(n)$ bits to Bob, or Bob has to send $\Omega(m)$ bits to Alice. 

To construct the reduction, we observe that any edge $e = (u,v)$ defines a partitioning of the compute nodes in the tree $G$
into two subsets: $V_{e}^-$ and $V_e^+$. Here, $V_e^-$ is the set of compute nodes in the same side as $u$, and $V_e^+$ in the same side as $v$. Hence, any algorithm that computes the set intersection in the tree topology also solves a lopsided set disjointness problem, where Alice holds all data located in $V_e^-$, Bob holds all data located in $V_e^+$, and they can only communicate through the edge $e$. Following this core idea, we can show the following lower bound.

\begin{theorem} \label{thm:lb}
Let $G=(V,E)$ be a symmetric tree topology.
Any algorithm computing the intersection $R \cap S$ has cost $\Omega(C_{LB})$, where
\[C_{LB} = \max_{e \in E} \frac{1}{w_e} \cdot \min \left \{|R|, |S|, \sum_{v \in V_e^-} N_v, \sum_{v \in V_e^+} N_v \right\}. \]
\end{theorem}

Observe that the above lower bound holds independent of the number of rounds that the algorithm uses.

\begin{proof}
Consider an edge $e \in E$. Any algorithm that computes the set intersection  $R \cap S$ must solve the following
problem. Alice holds two sets, $R_A = \bigcup_{v \in V_e^-} R_v$, and $S_A = \bigcup_{v \in V_e^-} S_v$.
Similarly, Bob holds two sets, $R_B = \bigcup_{v \in V_e^+} R_v$, and $S_B = \bigcup_{v \in V_e^+} S_v$.
Then, Alice and Bob must together compute two set intersections, $R_A \cap S_B$ and $R_B \cap S_A$,
communicating only through the link $e$ with bandwidth $w_e$. 
The lower bound for lopsided disjointness tells us that in order to compute $R_A \cap S_B$ we need to
communicate $\Omega(\min \{|R_A|, |S_B| \})$ bits, and for $R_B\cap S_A$ 
we need at least $\Omega(\min \{|R_B|, |S_A|\})$ bits.
Hence, the cost of any algorithm must be $\Omega(C)$, where:
\begin{align*}
C & =  \frac{1}{w_e} \max(\min \{|R_A|, |S_B| \}, \min \{|R_B|, |S_A|\} )\\
& \geq \frac{1}{2 w_e}  \min\{|R_A| + |R_B|, |S_A| + |S_B|, |R_A| + |S_A|,  |R_B| + |S_B|\} \\
& = \frac{1}{2 w_e} \min \left( |R|, |S|, \sum_{v \in V_e^-} N_v, \sum_{v \in V_e^+} N_v \right)
\end{align*}
Applying the above argument to every edge in the tree $G$, we obtain the desired result.
\end{proof}

\subsection{Warmup on Symmetric Star}

We first consider the star topology to present some of the key ideas. W.l.o.g. we assume $|R| \leq |S|$.
We present a one-round algorithm that is based on randomized hashing.  

Our algorithm (Algorithm~\ref{alg:set-intersection-star}) in its core performs a randomized hash join.
It first partitions the compute nodes into two subsets, $V_\alpha$ and $V_\beta$, depending on the size of the local data.
Define $N' = |R| + \sum_{v \in V_\alpha} |S_v|$.
Let $h$ be a random hash function that maps independently each
 $a$ in the domain to node $v \in V_C$ with the following probability:
\begin{align*}
 Pr[h(a) = v] = \begin{cases}
N_v / N', & v \in V_\alpha \\ 
|R_v| / N', & v \in V_\beta 
\end{cases}
\end{align*}
%\xiao{Note that $h$ doesn't need to be a truly random hashing function; instead a $O(\log N)$-wise independent hashing function suffices, which only requires $O(\log N)$ space to be stored. The performance of such a hash function will be analyzed later.}

If $V_\beta = \emptyset$, then the algorithm performs a distributed hash join using the above hash function $h$.
Observe that the algorithm does not hash each value uniformly across the compute nodes, but with probability
proportional to the input data $N_v$ that each node holds.

If $V_\beta \neq \emptyset$, we perform hashing only on a subset of the data using a subset of the nodes. 
In particular, each node $v \in V_\beta$ first gathers all the elements from $R$ (the smallest relation) and locally computes
 $R \cap S_v$, while hashing is used to compute the remaining set intersection.
After the data is communicated, the intersection can be computed locally at each node.

\begin{algorithm}
\caption{{\sc StarIntersect}$(G,\dstr)$}
\label{alg:set-intersection-star}
$V_\alpha \gets  \{v \in V_C \mid \min \{ N_v, N-N_v \} < |R| \}$, $V_\beta \gets V_C \setminus V_\alpha$ \;
\For{$v \in V_C$}{
		  send every $a \in R_v^\dstr$ to all nodes in $V_\beta \cup \{ h(a) \}$ \;
		  \If {$v \in V_\alpha$}{
		    send every $a \in S_v^\dstr$ to $h(a)$ \; }
		  }
\end{algorithm}

We next show that the above algorithm is optimal within a polylogarithmic factor.

\begin{lemma}
Let $G = (V,E)$ be a symmetric star topology, and consider sets $R,S$ with  $N = |R| + |S|$.
Then, {\sc StarIntersect} computes the set intersection $R \cap S$ with cost $O(\log N \log |V|)$ away from the optimal solution with high probability.
\end{lemma}

\begin{proof}
The correctness of the algorithm is straightforward. We will next bound the cost of the algorithm. We will measure the cost using elements of the set; to translate to bits it suffices to add a $\log (N)$ factor which captures the number of bits necessary to represent each element.

To make the notation simpler, we will use $w_v$ to refer to the bandwidth $w_e$ of edge $e = (v,w)$, where $v \in V_C$ and $w$ is the central node of the star topology. We can now reformulate the lower bound from Theorem~\ref{thm:lb} as
\[C_{LB} = \max \left \{\max_{v \in V_\alpha} \frac{\min\{N_v, N-N_v\}}{w_v}, \max_{v \in V_\beta}\frac{|R|}{w_v} \right \}\]
We now distinguish two cases, depending on whether the edge is adjacent to a node in $V_\alpha$ or $V_\beta$.

\paragraph{Case 1: $v \in V_\beta$}
Consider the two edges $(v,w)$ and $(w,v)$. The number of tuples that will be sent through edge $(v,w)$ is $|R_v| \leq |R|$.
As for the tuples received, node $v$ will receive $|R| - |R_v|$ tuples from $R$, as well as some tuples from $S$ which are in expectation:
$ \frac{|R_v|}{N'} \cdot  \sum_{v \in V_\alpha} |S_v| \leq |R_v|$.
Thus, the cost incurred by edges adjacent to $V_\beta$ is:
$
\max_{v \in V_\beta} \frac{|R|}{w_v} \leq C_{LB}
$.
Even though the above analysis just bounds the expectation, we can use Chernoff bounds to show that with probability
polynomially small in the number of compute nodes, the number of tuples will not exceed the expectation by more than an
$O(\log |V|)$ factor for any of the edges.
 
\paragraph{Case 2: $v \in V_\alpha$} We bound separately the number of $R$-tuples and $S$-tuples that go
through each edge.

The expected number of $S$-tuples that go through edge $(w,v)$ is
\begin{align*} 
\left(\sum_{u \in V_\alpha} |S_u| - |S_v|\right) \cdot \frac{N_v}{N'} & \leq {(N'-R- |S_v|) \cdot \frac{N_v}{N'}}   {\leq \frac{ (N' - N_v) N_v}{N'} } \le  \min \{N_v, N - N_v \}
\end{align*}
The third inequality is a direct application of the facts that $\min\{a,b\}  \ge \frac{a \cdot b}{a+b}$ for any $a,b \ge 0$ and $N' < N$.
Similarly, the expected number of $S$-tuples that go through edge $(v,w)$ is
\begin{align*} 
|S_v| \cdot \frac{N' - N_v}{N'}  \leq \frac{ (N' - N_v) N_v}{N'} \leq \min \{N_v, N - N_v \}
\end{align*}

For $R$-tuples, we distinguish two cases.
If $V_\beta = \emptyset$, then we can bound the expected size using the same argument as above for $S$-tuples.
We now turn to the case where $V_\beta \neq \emptyset$.

We first claim that $N_v \leq N-N_v$ for each vertex $v \in V_\alpha$.
Indeed, if not then we must have that $N- N_v < |R|$, which implies that $N_v > |S|$.
However, this is a contradiction since there exists $u \in V_\beta$ with $N_u > |R|$.
Hence, it suffices to bound the $R$-tuples that go through each edge by $N_v$. 

Indeed, the number of $R$-tuples that go through $(v,w)$ for $v \in V_\alpha$ are at most $|R_v| \leq N_v$.
As for the edge $(w,v)$, the expected number of tuples that use the edge is:
\begin{align*}
(|R| - |R_v|)  \cdot \frac{N_v}{N'} \leq \frac{|R|}{N'} \cdot N_v  \leq N_v
\end{align*}

Combining these two cases yields the desired claim. 
Note that all expectation calculations can be extended to high probability statements by losing a factor of $O(\log |V|)$ as mentioned before. 
\end{proof}

%%%%%%%%%%%%%%%%%%%%%%%%%%%%
\subsection{Algorithm on General Symmetric Tree}
\label{sec:set-intersection-tree}

We now generalize the algorithm for the star topology to an arbitrary (symmetric) tree topology.
W.l.o.g. we assume $|R| \leq |S|$.  %\xiao{Also, there is no data duplication in any input set.} 
We partition all edges in $E$ into two subsets: 
\begin{align*}
E_\alpha & =  \{e \in E \mid  \min\{\sum_{v \in V_{e}^+} N_v, \sum_{v \in  V_e^-}N_v\} < |R| \} \\
E_\beta & =  \{e \in E \mid \min\{\sum_{v \in V_e^+} N_v, \sum_{v \in  V_e^-}N_v\} \ge |R| \} 
\end{align*}

An edge $e \in E$ is called $\alpha$-edge if $e \in E_\alpha$, and $\beta$-edge if $e \in E_\beta$.
Observe that the definition is symmetric w.r.t. the direction of the edge: if $(u,v)$ is an $\alpha$-edge, so is $(v,u)$. The intuition behind this partition lies in  the lower bound of Theorem~\ref{thm:lb}, where the amount of data that can go through an $\alpha$-edge is $O(\min\{\sum_{v \in V_{e}^+} N_v, \sum_{v \in  V_e^-}N_v\})$ and through a $\beta$-edge is $O(|R|)$. We denote by $G_\beta$ the edge-induced subgraph of the edge set $E_\beta$.

\begin{lemma}
	\label{lem:beta-tree}
The subgraph $G_\beta$ is a connected tree.
\end{lemma}

\begin{proof}
For the sake of contradiction, assume there exist vertices $u,v \in V(G_\beta)$ such that $u,v$ are not connected in $G_\beta$. Then, there exists an $\alpha$-edge $e$ on the unique path that connects $u$ and $v$ in $G$. 
In turn, $e$ splits $G$ into two connected subtrees: $G_e^+$ (that contains all nodes in $V_e^+$), and $G_e^-$ (that contains all nodes in $V_e^-$). 
Suppose w.l.o.g. that $u \in V(G_e^+)$ and $v \in V(G_e^-)$. 

Since $u,v$ belong in the edge-induced subgraph of $E_\beta$, 
there exists $\beta$-edges $e_1 \in G_e^+, e_2 \in G_e^-$.
 We observe that $V_{e_1}^+ \subseteq V_e^+$ and $V_{e_2}^- \subseteq V_e^-$, which implies 
$|R| \le  \sum_{v \in V_e^+} N_v$ and $|R| \le \sum_{v \in V_e^-} N_v$.
In this way, $e$ would be an $\beta$-edge, contradicting our assumption.
\end{proof}

On the other hand, the edge-induced subgraph $G_\alpha$ derived from $E_\alpha$ is not necessarily connected and forms a forest.

\paragraph{Balanced Partition} The first step of our algorithm is to compute a partition $\{ V_C^1, V_C^2, \cdots, V_C^k\}$ of the compute nodes $V_C$.
In particular, the algorithm seeks a {\em balanced partition}, as illustrated in Figure~\ref{fig:balanced-partition}.

\begin{figure}
	\centering
	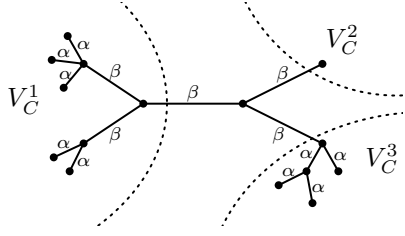
	\caption{An illustration of a balanced partition.}
	\label{fig:balanced-partition}
\end{figure}

\begin{definition}
\label{def:balanced-partition}
A partition $\{V_C^1, V_C^2, \cdots, V_C^k\}$ of the compute nodes $V_C$ is balanced for data distribution $\dstr$
if the following properties hold:
\begin{packed_enum}
\item If two nodes are connected in $G_\alpha$, they belong in the same block of the partition ;
\item Each edge appears in the spanning tree of at most one block of the partition ;
\item For every block $i$, $\sum_{v \in V_C^i} N^\dstr_v \geq |R|$ ;
\item For every $\beta$-edge $e$ in the spanning tree of a block $i$, 
 $\min \{ \sum_{v \in V_C^i \cap V_e^+} N_v, \sum_{v \in V_C^i \cap V_e^-} N_v \} \le |R|$.
\end{packed_enum}
\end{definition} 

Before we show how to find a balanced partition, we first discuss how we can use it to compute the set intersection.

\paragraph{The Algorithm} 
Let $\{V_C^1, V_C^2, \cdots, V_C^k\}$ be a balanced partition.
For every block $V_C^i$, we define a random hash function $h^i$ that maps independently each
value $a$ in the domain to node $v \in V_C^i$ with probability:
\[ Pr[h^i(a) = v] = \frac{N_v}{ \sum_{u \in V_C^i} N_u} \]
Using the above probabilities, we can now describe the detailed algorithm (Algorithm~\ref{alg:set-intersection-tree}), which works in a single round.  Each $R$-tuple is hashed across all blocks of the partition (hence it may be replicated), while each $S$-tuple is hashed only in the block that contains the node it belongs in. After all data is communicated, each node locally computes the set intersection. 
\begin{algorithm}
	\caption{{\sc TreeIntersect}$(G,\dstr)$}
	\label{alg:set-intersection-tree}
	
	Find a balanced partition $\{V_C^1, V_C^2, \cdots, V_C^k\}$\;
	\For{$v \in V_C$}{
		\For {$i=1, \dots, k$}{
		  send every $a \in R_v^\dstr$ to $h^{i}(a)$ \;
		  \If {$v \in V_C^i$}{
		    send every $a \in S_v^\dstr$ to $h^{i}(a)$ \; }
		  }
	}
\end{algorithm}

\begin{theorem} \label{thm:tree:intersect}
	On a symmetric tree topology $G=(V,E)$, the set intersection $R \cap S$ with $|R| + |S| = N$ can be computed in a single round with cost $O(\log N \log |V|)$ away from the optimal solution with high probability.
\end{theorem}

\begin{proof}
	The correctness of the algorithm comes from the fact that each subset of nodes $V_C^i$ computes $R \cap \bigcup_{v \in V_C^i} S_v$. Since $S = \bigcup_{i=1}^k \bigcup_{v \in V_C^i} S_v$, it follows that the algorithm computes all results in $R \cap S$. 
	
	We next analyze the cost. As before, we will measure the cost in number of tuples, and then pay a $O(\log N)$ factor to translate to bits. We first rewrite the lower bound as:
	\[ C_{LB} = \max \left \{ \max_{e \in E_\alpha} \frac{1}{w_e} \min\{\sum_{v \in V_e^+} N_v, \sum_{v \in V_e^-} N_v\}, \max_{e \in E_\beta} \frac{|R| }{w_e}\right\} \]
	
	We analyze the cost for the edges in $E_\alpha, E_\beta$ separately. 
	
	\paragraph{Case: $e \in E_\beta$} We will bound the amount of data that goes through $e$ by $O(|R|)$. The $R$-tuples that go through $e$ are at most $|R|$, so it suffices to bound the number of $S$-tuples that cross edge $e$. By property (2) of a balanced partition, $e$ is included in at most one spanning tree, say of block $V_C^i$.
	Then, w.h.p. the expected amount of $S$-tuples that goes through $e$ is at most 
	\begin{align*}
	&\frac{1}{\sum_{v \in V_C^i} N_v} \cdot \left(\sum_{v \in V_C^i \cap V_e^-} N_v \right) \cdot \left( \sum_{v \in V_C^i \cap V_e^+} N_v \right) 
	\leq  \min \left\{  \sum_{v \in V_C^i \cap V_e^-} N_v, \sum_{v \in V_C^i \cap V_e^+} N_v \right\} \leq  |R|
	\end{align*}
	The first inequality comes from the fact that $\frac{a \cdot b}{a+b} \le \min\{a,b\}$ for any $a, b > 0$. The
	second inequality is implied directly by property (4) of a balanced partition.
	
	\paragraph{Case: $e \in E_\alpha$} We will bound the amount of data that goes through $e$ by 
	$\min \left \{\sum_{v \in V_e^-} N_v, \sum_{v \in V_e^+} N_v \right \}$.
	To bound the number of $S$-tuples, we again notice that $e$ can belong in the spanning tree of at most one block,
	say $V_C^i$. Hence, as in the previous case, w.h.p. the expected amount of $S$-tuples that goes through $e$ is at most 
	\begin{align*}
	& \frac{1}{\sum_{v \in V_C^i} N_v} \cdot \left(\sum_{v \in V_C^i \cap V_e^-} N_v \right ) \cdot  \left( \sum_{v \in V_C^i \cap V_e^+} N_v \right) 
	\leq  \min  \left \{  \sum_{v \in V_C^i \cap V_e^-} N_v, \sum_{v \in V_C^i \cap V_e^+} N_v\right\} 
	\leq \min  \left \{  \sum_{v \in V_e^-} N_v, \sum_{v \in V_e^+} N_v \right\} 
	\end{align*}
	We can bound the number of $R$-tuples that go through $e$ by distinguishing three cases:
	\begin{itemize}
		\item none of $G_e^-, G_e^+$ contain $\beta$-edges. Then, the partition consists of a single block, and the number of $R$-tuples can be bounded as we did above with the $S$-tuples.
		\item $G_e^+$ contains $\beta$-edges but $G_e^-$ not. Then, all vertices in $G_\beta$ are in $V_e^+$. The $R$-data that goes through $e$ is sent by nodes in $V_e^-$, so its size is bounded by $\sum_{v \in V_e^-} |R_v| \le  \sum_{v \in V_e^-} N_v = \min \left \{\sum_{v \in V_e^-} N_v, \sum_{v \in V_e^+} N_v \right \}$. Here, the last equality follows from the fact that $G_e^+$ contains at least one $\beta$-edge, which implies that
		$\sum_{v \in V_e^+} N_v \geq |R| > \sum_{v \in V_e^-} N_v$.
		\item $G_e^-$ contains $\beta$-edges but $G_e^+$ not. Then, all nodes in $V_e^+$ belong in the same block $V_C^i$. We can abound the expected amount of $S$-tuples with:
		\begin{align*}
		\ \ \ \ \ \ \ \ \ \ \ \ \ & \frac{1}{\sum_{v \in V_C^i} N_v} \cdot \left( \sum_{v \in V_e^-} |R_v| \right) \cdot \left ( \sum_{v \in V_C^i \cap V_e^+} N_v \right) \\
		\leq & \frac{ \sum_{v \in V_e^-} |R_v| + \sum_{v \in V_C^i \cap V_e^+} N_v }{\sum_{v \in V_C^i} N_v}  \min \left \{  \sum_{v \in V_e^-} |R_v|, \sum_{v \in V_C^i \cap V_e^+} N_v \right \} \\
		\leq & \frac{ |R| + \sum_{v \in V_C^i } N_v }{\sum_{v \in V_C^i} N_v}  \min \left \{  \sum_{v \in V_e^-} N_v, \sum_{v \in V_e^+} N_v \right \} 
		\leq  2  \min \left\{  \sum_{v \in V_e^-} N_v, \sum_{v \in V_e^+} N_v \right \}
		\end{align*}
		where the last inequality is from property (3) of Definition~\ref{def:balanced-partition}.
	\end{itemize}
	
	This completes the proof.
\end{proof}

%In Appendix~\ref{appendix:tree-set-intersect}, we show that the algorithm (almost) matches our lower bound, thus completing the proof for Theorem~\ref{thm:tree:intersect}.

\paragraph{Finding a Balanced Partition} 
Finally, we present how we can compute a balanced partition \xiao{in Algorithm~\ref{alg:balanced-partition}.} We say that two vertices in $G$ are $\alpha$-connected
if there exists a path that uses only $\alpha$-edges that connects them. 
For the algorithm below, denote $\Gamma(x)$ as the set of nodes that are $\alpha$-connected with node $x$ in $G$ \xiao{(line 2)}. Moreover, we use $w(x)$ to denote
the quantity $\sum_{x \in \Gamma(x)} N_x$, i.e. the total amount of data in the nodes from $\Gamma(x)$. 

\begin{algorithm}
	\caption{{\sc BalancedPartition$(G, \dstr)$}}
		\label{alg:balanced-partition}

	\For{$x \in V(G_\beta)$}{
		$\Gamma(x) \gets \{ v \in V_C \mid v,x  \text{ are } \text{$\alpha$-connected in } G \}$  \;}
	$\mathcal{P} \gets \emptyset$ \;
	\While{$|V(G_\beta)| > 0$}{
		pick the leaf vertex $x \in G_\beta$ with the smallest $w(x)$\;
		\uIf{$w(x) \ge |R|$}{
			add $\Gamma(x)$ to $\mathcal{P}$\;
		}
		\Else{
			$y \gets$ unique neighbor of $x$ in $G_\beta$\;
			$\Gamma(y) \gets \Gamma(y) \cup \Gamma(x)$\; 
		}
		$G_\beta \gets G_\beta \setminus \{x\}$\;
	}
	\Return $\mathcal{P}$ \;
\end{algorithm}

The algorithm initially creates a group for each set of compute nodes that are connected through $\alpha$-edges.
Then, it starts merging the groups (starting from the leaves of the tree) as long as the total number of the elements in the group
is less than $|R|$. We show %(in Appendix~\ref{appendix:balanced}) 
that the above algorithm indeed creates the desired balanced partition.

\begin{lemma} \label{lem:balanced}
{%On a symmetric tree topology $G=(V,E)$, and an initial data distribution $\dstr$,  
	Algorithm~\ref{alg:balanced-partition} outputs a balanced partition of compute nodes $V_C$ in $O(|V|)$ time.}
\end{lemma}

\begin{proof}
	First, we notice that in lines 1-2 each compute node $V_C$ belongs in exactly one $\Gamma(x)$. In the remaining algorithm, every vertex in $G_\beta$ with $w(x) > 0$ is put into exactly one block, thus $\mathcal{P}$ is a partition of $V_C$.  Indeed, the only issue may occur when we are left with a single vertex $x$: we claim that in this case we always have $w(x) \geq |R|$. Suppose $w(x) < |R|$, and consider the last vertex $u$ for which $\Gamma(u)$ was added in $\mathcal{P}$ (such a vertex always exists, since every leaf vertex of $G_\beta$ initially has weight at least $|R|$). But then, the algorithm could not have picked $u$ at this point, since all other leaf vertices have smaller weight, a contradiction.
	
	We now prove that the output partition satisfies all properties of a balanced partition (Definition~\ref{def:balanced-partition}).
	
	(1) The first condition is trivial. From lines 1-2, two compute nodes that are connected in $G_\alpha$ will be in
	the same initial $\Gamma(x)$, hence they will appear together in a block of the partition. 
	
	(2) By contradiction, assume there exists an edge $e = (u,v)$ appearing in the spanning trees of $V_C^i$ and $V_C^j$ for $i \neq j$. By the definition of spanning trees, there exists one pair of vertices $x,y \in V_C^i$ and one pair of vertices $x',y' \in V_C^j$ such that $x,x' \in G_e^+$ and $y, y' \in G_e^-$.  When Algorithm~\ref{alg:balanced-partition} visits $e$ in line 9, w.l.o.g. assume $u$ is visited before $v$. Since $x,x'$ are placed in different blocks of the partition, it cannot be that both $x, x' \in \Gamma(u)$. W.l.o.g., $x' \notin \Gamma(u)$. This implies that $x'$ has already been put into one block with  vertices from $G_e^-$. Then  $x', y'$ won't appear in the same block, contradicting our assumption.
	
	(3) It is easy to see that the algorithm adds a set of nodes to $\mathcal{P}$ only if their total weight is at least $|R|$.
	
	(4) Consider a block $V_C^i$ in the partition. Let $e = (u,v)$ be a $\beta$-edge in the spanning tree of $V_C^i$. Then,
	Algorithm~\ref{alg:balanced-partition} visits $e$ in line 9: w.l.o.g. assume $u$ is visited before $v$. At this point, we 
	have $w(u) < |R|$, since $\Gamma(u)$ was merged with $\Gamma(v)$. The key observation is that we have $\Gamma(u) = V_C^i \cap V_e^-$, since no other compute nodes will be added to the "left" of $e$ (since $u$ is a leaf node). Hence, 
	\begin{align*}
	\min \{ \sum_{v \in V_C^i \cap V_e^+} N_v, \sum_{v \in V_C^i \cap V_e^-} N_v \} \leq \sum_{v \in V_C^i \cap V_e^-} N_v
	= w(u) < |R|
	\end{align*}
	This completes the proof.	
\end{proof}

\paragraph{Remark} Interestingly, the algorithm we described above does not use the link bandwidths to decide what to send and where to send to. Instead, what matters is the connectivity of the network and how the data is initially partitioned across the compute nodes. This is a significant practical advantage because bandwidth information may be imprecise or have high variability at runtime, such as when sharing a cluster with other users.

	\section{Cartesian Product}
\label{sec:cartesian-product}

In the cartesian product problem, we are given two sets $R, S$ with $|R| = |S|= N/2$. (We will discuss in the end why the unequal case is challenging, even on the simple symmetric star topology).
Our goal is to enumerate all pairs $(r,s)$ for any $r \in R, s \in S$, such that the output pairs are distributed among the compute nodes by the end of the algorithm. {Similar to set intersection, there is no designated node for each output pair, as long as it is emitted by at least one node.}
We assume that all elements are drawn from the same domain, and that initially the input data is partitioned across the compute nodes.

%%We first consider the problem for the case of a symmetric star topology and then extend it to the general tree topology, where the cost function  is $f_e(x) = x/w_e$. 

\subsection{Lower Bounds on Symmetric Trees}

We present two lower bounds on cost for the case of a symmetric tree topology.  The first one as stated in Theorem~\ref{thm:lb-cp1} has the same form as the one in Theorem~\ref{thm:lb} when $|R| = |S| = N/2$, but uses a slightly different argument. Both lower bounds are expressed in terms of elements, and not bits. 

\begin{theorem}
	\label{thm:lb-cp1}
	Let $G=(V,E)$ be a symmetric tree topology.
	Any algorithm computing $R \times S$ has (tuple) cost $\Omega(C_{LB})$, where
	\[C_{LB} = \max_{e \in E} \frac{1}{w_e} \cdot \min \left \{\sum_{v \in V_e^-} N_v, \sum_{v \in V_e^+} N_v \right\}. \]
\end{theorem}

\begin{proof}
	Let $\copt$ be the cost of any algorithm computing $R \times S$ on the tree topology $G$. Consider an edge $e \in E$. Suppose that $C_{opt} \cdot w_e \le \sum_{v \in V_e^-} |R_v|$. Then, at least one element in $R_u$ for some $u \in V_e^-$ does not go through $e$, i.e., entering into any vertex in $V_e^+$. In this case, in order to guarantee correctness, all data in $S$ must be sent to $u$, hence $\copt \cdot w_e \ge \sum_{v \in V_e^+} |S_v|$. Thus $\copt \cdot w_e \ge \min\{\sum_{v \in V_e^-} |R_v|, \sum_{v \in V_e^+} |S_v|\}$. Using a symmetric argument, $\copt  \cdot w_e \ge  \min \{ \sum_{v \in V_e^-} |S_v|, \sum_{v \in V_e^+} |R_v|\}$. Summing up the two inequalities, and 
	observing that $\min \{\sum_{v \in V_e^-} N_v, \sum_{v \in V_e^+} N_v\} \leq |R| (=|S| = N/2)$, we obtain the lower bound on edge $e$.

	Applying the above argument to every edge in the tree $G$, we obtain the desired result.
\end{proof}

The second lower bound uses a different argument that depends on the underlying tree topology. To state the lower bound, we first define a "directed" version $G^\dagger$ of the symmetric tree $G$ as follows. $G^\dagger$ has the same vertex set as $G$. Recall that each edge $e = (u,v)$ in $G$ partitions the nodes of $V$ into $V_e^+$ and $V_e^-$. Then, if $\sum_{x \in V_e^-} N_x \leq \sum_{x \in V_e^+} N_x $, $G^\dagger$ contains only an edge from $u$ to $v$, otherwise only an edge from $v$ to $u$. As the next lemma shows, the resulting directed graph $G^\dagger$ has a very specific structure. 
%{in order for the proof to work, we have to make sure that $G$ contains no vertices with degree 2, which we can achieve without any loss of generality.}

\begin{lemma}
\label{lem:directed-alpha}
$G^\dagger$ satisfies the following properties:
\begin{packed_enum}
\item The out-degree of every node is at most one.
\item There exists exactly one node with out-degree zero.	
\end{packed_enum}
\end{lemma}

\begin{proof}
	By contradiction, assume there exists one node $u \in V$ with at least two out-going edges. Since $G$ has no vertices with degree 2, this means that 
	 $G^\dagger$ has three edges ${e_1} = (u, v_1)$, ${e_2} = (u, v_2)$, ${e_3} = (u, v_3)$. For each such edge, we have $\sum_{x \in V^{+}_{e_i}} N_x \geq \sum_{x \in V^{-}_{e_i}} N_x$, and thus $\sum_{x \in V^{+}_{e_i}} N_x \geq N/2$. Observe that because $G$ is a tree, it also holds that the vertex sets $V^+_{e_i}$ are disjoint. 
 Then we come to the following contradiction
$$ N = \sum_{x \in V} N_x \geq \sum_{i=1}^3 \sum_{x \in V^{+}_{e_i}} N_x  \geq 3N/2$$
	thus (1) is proved. 
	
Since $G^\dagger$ is a directed tree, it is easy to see that there must exist at least one node with no outgoing edges; otherwise, there would be a cycle in the graph, a contradiction.	Hence, it suffices to show that there is at most one such a node. By contradiction, assume two nodes $u,v$ with outdegree 0. 
 Consider the unique path between $u,v$: then, there must be a node in the path with out-degree at least two. However, this contradicts (1), thus (2) is proven as well.
\end{proof}

We denote the single node with out-degree zero as $r$, and call it the {\em root} of the tree. Every other node in $G^\dagger$ will point towards $r$, as the example in Figure~\ref{fig:directed-tree} illustrates. Observe that the root $r$ of the tree could be a compute node. But in this case, the algorithm that simply routes all the data to the root is asymptotically optimal, since the cost matches the cost of the lower bound in Theorem~\ref{thm:lb-cp1}. Hence, we will focus on the case where the root is not a compute node; in this case, it is easy to observe that all the nodes in $G^\dagger$ with in-degree 0 are exactly the compute nodes.

\begin{figure}
	\centering
	\includegraphics[scale=1.2]{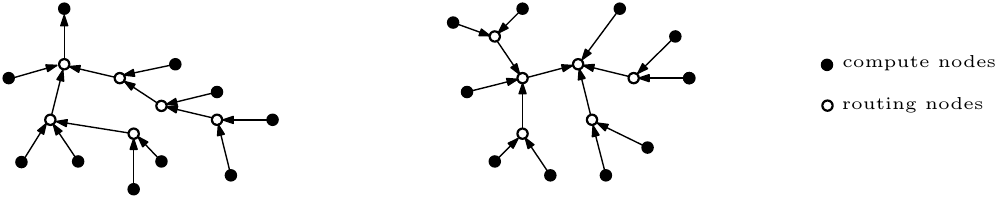}
	\caption{Two examples of a directed graph $G_\dagger$. The left one is rooted at a compute node and the right one is rooted at a router.  }
	\label{fig:directed-tree}
\end{figure}

A {\em cover} of $G^\dagger$ is a subset $S \subseteq V$ such that every leaf node has some ancestor in $S$. We will be interested in {\em minimal covers} of $G^\dagger$. Observe that the singleton set $\{r\}$ is trivially a minimal cover.   

\begin{theorem} \label{thm:lb-cp2}
	Let $G=(V,E)$ be a symmetric tree topology.
	Let $U$ be a minimal cover of $G^\dagger$ such that $U \neq \{r\}$, where $r$  
	is the root of $G^\dagger$. Then, any algorithm computing the cartesian product $R \times S$ for $|R| = |S| = N/2$ has (tuple) cost $\Omega(C_{LB})$, where
	\[C_{LB} =  \frac{N}{\sqrt{\sum_{v \in U} w^2_v}}, \]
	where $w_v$ is the capacity of the unique outgoing edge of $v$ in $G^\dagger$.
\end{theorem}

\begin{proof}

%If $U$ only contains a single node, say $U = \{u\}$, then $R'_u = R$ and $S'_u = S$ by Definition~\ref{def:new-distribution}. The solution is trivial by letting $u$ perform local computation. In this case, the cost is $0$ which is matched by the formula above with $w_v = + \infty$. 
	
%In general, $G_U$ is a connected subtree with $|U| \ge 2$. 
Let $e_u$ be the outgoing edge of $u \in U$ in $G^\dagger$, with capacity cost $w_u$. Let $T_u$ be the subtree rooted at $u$. From minimality of $U$, it follows that $T_u, T_v$ have disjoint vertex sets. Moreover, from the definition of a node cover, every compute node belongs in some (unique) subtree. This means that we can bound the output result by at most the union of the outputs in the compute nodes of each subtree. In the following, we will bound the maximum output size of a given subtree $T_u$.

Let $R_u',S_u'$ denote the elements of $R,S$ respectively that are in some compute node of $T_u$. Moreover, let $R_u'', S_u''$ be the elements of $R,S$ that go through link $e_u$ respectively.
	Then, the size of the results that can be produced at subtree $T_ u$ is at most $|R'_u \cup R''_u| \cdot |S'_u \cup S''_u|$.
	Observe the following:
	\begin{itemize}
		\item $|R''_u| \leq \copt  \cdot w_u$ and $|S''_u| \leq \copt  \cdot w_u$;
		\item $|R'_u| \leq \copt  \cdot w_u$ and $|S'_u| \leq \copt  \cdot w_u$.
		Indeed, since $w_u$ is an outgoing edge of $u$ in $G^\dagger$, Theorem~\ref{thm:lb-cp1} tells us that $\copt \cdot w_u\geq |R'_u| + |S'_u|$.
	\end{itemize}
	Hence, we can bound the number of outputs in $T_u$ as:
	\begin{align*}
		|R'_u \cup R''_u| \cdot |S'_u \cup S''_u| 
		& \leq (|R'_u| + |R''_{u}|) (|S'_u| + |S''_u|) \\
		& \leq (2 \cdot \copt  \cdot w_u) (2 \cdot \copt \cdot w_u) 
		 = 4 \cdot C_{opt}^2 \cdot w_u^2
	\end{align*}

	In order for the algorithm to be correct, the total size of the output must be at least $|R| \cdot |S|$. Summing over all nodes in the minimal cover $U$, we obtain $|R| \cdot |S |\geq 4 \cdot C_{opt}^2 \cdot \sum_{u \in U} w_u^2$.
%	\begin{align*}
%	|R| \cdot |S| \le &  \sum_{u \in U} |R'_u \cup R''_u| \cdot |S'_u \cup S''_u| \\
%	\leq & \sum_{u \in U} (|R'_u| + |R''_{u}|) (|S'_u| + |S''_u|) \\
%	\leq & \sum_{u \in U} (2 \cdot \copt  \cdot w_u) (2 \cdot \copt \cdot w_u) 
%	=  4 \cdot C_{opt}^2 \cdot \sum_{v \in V_C} w_v^2
%	\end{align*}
	This concludes the proof.  
\end{proof}

\subsection{The Weighted HyperCube Algorithm}
\label{sec:whc}

 In this section, we present a deterministic one-round protocol on a symmetric star topology, named {\em weighted HyperCube} (wHC), which generalizes the HyperCube algorithm~\cite{afrati2011optimizing}. We assume that the data statistics $|R_v|$, $|S_v|$ are known to all compute nodes.

We give a strict ordering $\leq$ on the compute nodes in $V_C$. Each node assigns consecutive numbers to its local data. More specifically, node $v$ labels its data in $R_v$ from $1 + \sum_{u < v} |R_u| $ to $ \sum_{u \le v} |R_u|$, and data in $S_v$ from $1+ \sum_{u < v} |S_u|$ to $ \sum_{u \le v} |S_u|$. In this way, each element from $R$ is labeled with a unique index, as well as each one from $S$. In this way, each answer in the cartesian product can be uniquely mapped to a point in the grid $\square = \{1, 2, \dots, |R| \} \times \{1, 2, \dots, |S| \}$.

The wHC protocol assigns to each compute node $v$ a square $\square_v$ centered at  $(x_v, y_v)$ with dimensions $l_v \times l_v$. Then, a tuple $r_i \in R$ will be sent to $v$ if $ x_v - l_v \le i \le x_v + l_v$, and a tuple $s_j \in S$ will be sent to $v$ if $y_v - l_v \le j \le y_v + l_v$. After all tuples are routed, the cartesian product will be computed locally at each compute node. To guarantee correctness, we have to make sure that $\bigcup_v \square_v = \square$, i.e., the squares assigned to each node fully cover the grid.

We first compute the dimensions $l_v$ of the square assigned to each node. Intuitively, we want to make sure that $l_v$ is proportional to the capacity of the link. However, to make sure that we can {\em pack} the resulting set of squares without any overlap, we consider squares that are powers of 2. Specifically, 
\begin{align}
 l_v = 	\arg \min_k \{ 2^k \geq w_v \cdot L\}, \qquad \qquad L =  \frac{N}{\sqrt{\sum_u w_u^2}}
\end{align}

Second, we need to specify the positions of the squares, i.e. determine how they can be {\em packed} without any overlap. An example of such a packing is given in Figure~\ref{fig:whp}.
To pack the squares, we will make use of the following lemma.

\begin{lemma}[Packing Squares]
\label{lem:square:packing}
Let $S$ be a set of squares $d_i \times d_i$, where each $d_i$ is a power of two. Then, we can pack the squares in $S$ such that they fully cover a square of size at least  $1/2 \sqrt{ \sum_i d_i^2}$. 	
\end{lemma}

\begin{proof}
We provide an algorithm for the packing. We start the following procedure in an increasing order of $i \geq 0$: for each $i$, if there are $4$ squares of size $2^i \times 2^i$ in $S$, we pack them into a larger square of size $2^{i+1} \times 2^{i+1}$. In this way, we can transform $S$ into a new set of squares $S'$, where for every $i$, there are at most $3$ squares of size $2^i \times 2^i$. It is now easy to see that, by induction starting from the smaller size, all squares of size $\leq 2^{i-1}$ can be packed inside a square of size $2^{i}$. 
Hence, we can pack all squares in $S'$ inside a square of size $2^{i^*+1}$, where $2^{i^*}$ is the dimension of the largest square in $S'$. To conclude the argument, observe that the square with dimension $2^{i^*}$ is fully packed. Also, $2^{i^*+1} \geq \sqrt{\sum_i d_i^2}$. Hence, we can fully pack a square of size at least $1/2 \sqrt{\sum_i d_i^2}$. 
\end{proof}

%\begin{algorithm}
%	\caption{{\sc BalancedPackingStar}$(o, V_C, \square, w(.))$}
%	\label{alg:balanced-packing}
%	
%	$L \gets  N/2 \sqrt{\sum_v w_v^2}$ \;
%	$w \gets \max_{v \in V_C} w_v$\;
%	\While{$\square$ is not fully covered}{
%		$u \gets \arg \max_{v \in V_C} w_v$\;
%		$\ell \gets \arg \min_k \{ w \geq 2^k \cdot w_u \}$\;
%		Assign $\square_u$ to $u$ of size $(2^{-\ell} wL) \times (2^{-\ell}wL)$\;
%		$V_C \gets V_C - \{u\}$\;
%	}
%\end{algorithm}

%We next describe an algorithm called {\sc BalancePackingStar} that achieves this requirement.Define $L = N/ 2 \sqrt{ \sum_v w_v^2}$, and let $w$ be the largest capacity in the network.

%Since $w \cdot L \leq N/2$, we can assume that $(w \cdot L)$ is a multiple of $|R| = |S| = N/2$ (losing only a constant factor).
%The algorithm goes over the nodes in decreasing order of capacity, and essentially assigns to each a square where its side is 
%proportional\ to the capacity.
%However, to make sure that the squares can cover the grid such that they do not overlap, we consider squares of sizes that decrease as powers of 2. An example of such a packing is given in Figure~\ref{fig:whp}.

\begin{figure}
	\centering
	\includegraphics[scale=0.45]{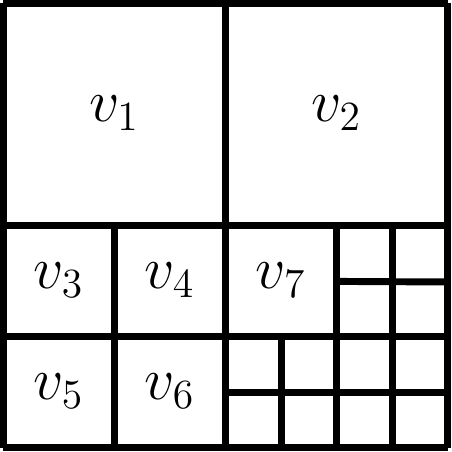}
	\caption{An example of packing squares.}
	\label{fig:whp}
\end{figure}

The next lemma bounds the cost of the wHC protocol.

\begin{lemma} \label{lem:whc}
Let $G$ be a symmetric star topology. Then, the wHC algorithm correctly computes the cartesian product $R \times S$ for $|R| = |S| = N/2$  with (tuple) cost $O(C)$, where
\[ C =  \max \left\{ \max_v \frac{N_v}{w_v}, \frac{N}{\sqrt{\sum_v w_v^2}} \right\}\]
\end{lemma}

\begin{proof}
%To prove correctness, it suffices to show that the grid is fully covered.
%Indeed, notice that each node $v$ covers an area of size at least $(L \cdot w_v)^2$.
%Thus, summing over all compute nodes, the area covered in total is at least
%\[\sum_{v \in V_C} (L \cdot w_v)^2 = (N/2)^2 = |R| \cdot |S| \]
%Hence, the whole area of $\square$ is covered.
To prove correctness, we apply Lemma~\ref{lem:square:packing} with $S = \{\l_v \times l_v \mid v \in V_C \}$. Then, the squares fully pack a square of area at least 
\[ \frac{1}{4} \sum_{v \in V_C} (2^{l_v})^2 \geq \frac{1}{4} \sum_{v \in V_C} (L \cdot w_v)^2  = (N/2)^2 = |R| \cdot |S| \]
Hence, the whole grid can be covered.

Next, we analyze the cost of the algorithm. First, the cost of sending data is $\max_v {N_v/w_v}$.
For the cost of receiving, observe that node $v$ receives at most $2 \cdot (2 L \cdot w_v) = 4 w_v L$ tuples.
Hence, the cost of receiving is bounded by $4 L$. Combining these two costs obtains the desired result.
\end{proof}

\subsection{Warm-up on Symmetric Star}

Before we present the general algorithm for symmetric trees, we warm up by studying the simpler symmetric star case (Algorithm~\ref{alg:cartesian-lower-bound}). 

The algorithm checks whether the maximum data that some node holds exceeds $N/2$. If so, it is easy to observe that the strategy where every compute node sends their data to that node is optimal. If every node holds at most $N/2$ data initially, then in $G^\dagger$ all compute nodes of the star are directed to the central node $o$, which becomes the root of $G^\dagger$. In this case, running the wHC algorithm on the whole topology can be proven optimal.

\begin{algorithm}
	\caption{{\sc StarCartesianProduct}$(G, \dstr)$}
	\label{alg:cartesian-lower-bound}
	
	\If{$\max_u N_u > N/2$}{
		all compute nodes send their data to $\arg\max_u N_u$\;
	}
	\Else{
		run the wHC algorithm \;
	}
\end{algorithm}

%%We present the detailed proof in Appendix~\ref{appendix:star:optimal}.

\begin{lemma} \label{lem:star:cp:optimal}
On a symmetric star topology, the {\sc StarCP} algorithm correctly computes the cartesian product $R \times S$ for $|R| = |S| = N/2$ in a single round deterministically and with cost $O(1)$ away from the optimal.
\end{lemma}

\begin{proof}[Proof of Lemma~\ref{lem:star:cp:optimal}]
	We distinguish the analysis into two cases, depending on whether $\max_u N_u > N/2$ or not.
	
	First, suppose that $\max_u N_u > N/2$. Let $u^* = \arg\max_u N_u$.
	For node $u^*$, we have $N- N_{u^*} < N/2 < N_{u^*}$.
	For every other node $v \neq u^*$,  it holds $N_v < N/2$, hence $N_v < N - N_v$.
	Hence, we can write Theorem~\ref{thm:lb-cp1} as:
	\[ C_{opt} \geq \max_v \frac{1}{2 w_v} \min \{N_v, N-N_v\} \ge \max \left\{ \frac{N-N_{u^*}}{2 w_{u^*}},  \max_{v \neq u^*} \frac{N_v}{2w_v} \right\} \]
	But this is exactly half the cost of the protocol where all nodes send their data to $u^*$.
	
	Suppose now that $\max_u N_u \leq N/2$. From Theorem~\ref{thm:lb-cp1}, we obtain the lower bound
	$C_{opt} \ge \max_{v} \frac{N_v}{2w_v}$. Additionally, observe that in $G^\dagger$ all compute nodes of the star are directed to the central node $o$, and hence $V_C$ is a minimal cover of $G^\dagger$. Indeed, if we add $\{o\}$ to $V_C$, the cover is not minimal, since $\{o\}$ is a minimal cover by itself. Plugging this cover in Theorem~\ref{thm:lb-cp2}, we obtain that $C_{opt} \ge N / \sum_v w_v^2$. To conclude, notice that these two lower bounds on $C_{opt}$ match the upper bound of wHC in 
	Lemma~\ref{lem:whc} within a constant factor.
\end{proof}

\subsection{Algorithm on Symmetric Tree}

We now generalize the techniques for the star topology to an arbitrary tree topology.

\paragraph{The Algorithm} Assume that the data statistics $|R_v|$, $|S_v|$ are known to all compute nodes. Similar to the wHC algorithm, each tuple from $R$ is labeled with a unique index, as well as each one from $S$. In this way, each answer in the cartesian product can be uniquely mapped to a point in the grid $\square = \{1, \dots, |R| \} \times \{1, \dots, |S| \}$. For simplicity, we split the routing phase into two steps.

Let $r$ be the root of the directed graph $G^\dagger$. In the first step, each compute node $v \in V_C$ sends its local data to $r$.

In the second step, we assign to each compute node $v \in V_C$ a square $\square_v$ such that every result $t = (t_r, t_s)$ is computed on some $v$. To compute $t$, associated tuples $t_r, t_s$ will be sent to $v$ at least once. In this step, every tuple sent to $v$ will be sent from the root $r$, which has gathered all necessary data in the first step. Next, we show how to find a balanced assignment on a tree and analyze its capacity cost with respect to the lower bound in Theorem~\ref{thm:lb-cp2}.

\paragraph{Balanced Packing on Symmetric Tree} 
%Given the input graph $G$, if it is a forest, we introduce a supernode $r$ as the root. For each root of the subtree in $G$, say $r'$, we add an edge $e = (r, r')$ with $w(e) = +\infty$.
Let $\zeta(u)$ be the set of children nodes of $u$ in $G^\dagger$, and $p_u$ the unique parent of $u$ in $G^\dagger$. To simplify notation, we use $w_v$ to denote the quantity $w(v,p_v)$.

The algorithm is split into two phases. First, it computes a quantity $\tilde{w}_v$ for each node $v$ in $G^\dagger$. For the leaf nodes, we have $\tilde{w}_v = w_v$, while for the internal nodes $\tilde{w}_v$ is computed in a bottom-up fashion (through a post-order traversal).
% Moreover, we generalize the notion $\square_u$ to every vertex in $G$, denoting the square allocated to $u$ with sizes $d_u \times d_u$.
%Note that the invariant of $x_u = \sqrt{\sum_{v \in y_u} w^2(v,p_v)}$ always holds. 
In the second phase, the algorithm computes a quantity $l_v$ for each node, but now in a top-down fashion (through a pre-order traversal).
%At the same time, it calls the procedure {\sc BalancedPackingStar} on all children nodes of $u$ to assign $\square_u$ to its children nodes $\zeta(u)$. 
As a final step, each compute node $v$ rounds up $(N/2)\cdot l_v$ to the closest power of 2, and then gets assigned a square of that dimension.

\begin{algorithm}	
	\caption{{\sc BalancedPackingTree}$(G)$}
	\label{alg:balanced-packing-tree}
	
	\ForAll{$v \in V \setminus \{r\}$ in post-order}{
		\If{$v$ is a leaf}{
			$\tilde{w}_v \gets w_v$\;
		}
		\Else{
		$\tilde{w}_v \gets  \min \{ w_v, \sqrt{\sum_{u \in \zeta(v)} \tilde{w}^2_u} \} $
		}
		}
	$\tilde{w}_r \gets \sqrt{\sum_{u \in \zeta(r)} \tilde{w}^2_u} $ \;	
	$l_r \gets 1 $\;
	\ForAll{$v \in V \setminus \{r\}$  in pre-order}{
 		 	$l_v \gets l_{p_v} \cdot  \tilde{w}_v / \sqrt{\sum_{u \in \zeta(p_v)} \tilde{w}^2_u}$\;
 	 	}
 	\ForAll{$v \in V_C$}{
 	  $d_v \gets \arg \min_k \{ 2^k \geq N \cdot l_v \}$\;
 	  assign to $v$ a square $d_v \times d_v$\;
 	}	
\end{algorithm}

The next lemma shows that Algorithm~\ref{alg:balanced-packing-tree} guarantees certain properties for the computed quantities.

\begin{lemma}
	\label{lem:properties}
	The following properties hold:
	\begin{packed_enum}
	\item For every non-root vertex $v$, $\tilde{w}_v \leq w_v$.
	\item For every vertex $v$, $l_v \leq \tilde{w}_v/\tilde{w}_r$.	
	\item There exists a minimal cover $U$ of $G^\dagger$ such that $\tilde{w}_r = \sqrt{\sum_{u \in U} w_u^2}$.
	\item For every vertex $u$, $l_u = \sqrt{\sum_{v \in T_u \cap V_C} l^2_v}$ where $T_u$ is the subtree rooted at $u$.
	\end{packed_enum}

\end{lemma}

\begin{proof}
Property (1) is straightforward from the algorithm. 
\smallskip

We prove property (2) by induction. For the base case, $v$ is the root. In this case, $l_r=1$, so the inequality holds with equality. 
Consider now any non-root vertex $v$ with parent $p_v$. We then have:
	\begin{align*}
	l_v = \frac{ l_{p_v} \cdot \tilde{w}_v}{\sqrt{\sum_{u \in \zeta(p_v)} \tilde{w}^2_{u}}} \le  
	\frac{\tilde{w}_v}{\tilde{w}_r} \cdot \frac{\tilde{w}_{p_v}}{\sqrt{\sum_{u \in \zeta(p_v)} \tilde{w}^2_{u}}} \le \frac{\tilde{w}_v}{\tilde{w}_r}
	\end{align*}
The first inequality holds from the inductive hypothesis for the parent node $p_v$. The second inequality comes from line 5 of the algorithm, which implies that $\tilde{w}_v \le \sqrt{\sum_{u \in \zeta(v)} \tilde{w}^2_{u}}$ for every non-leaf vertex $v$.
\smallskip

We also use induction to show property (3). For a subtree rooted at leaf node $v$, $U = \{v\}$ is a minimal cover. In this case, $\tilde{w}_v = w_v = \sqrt{\sum_{u \in U} w_u^2}$. For the induction step, consider some non-leaf node $v$. If $\tilde{w}_v = w_v$, then $\tilde{w}_v = \sqrt{\sum_{u \in U} w_u^2}$ holds for the minimal cover $U = \{v\}$. Otherwise, $\tilde{w}_v = \sqrt{\sum_{u \in \zeta(v)} \tilde{w}^2_u}$. From the induction hypothesis, there exists a minimal cover $U_u$ for the subtree rooted at $u \in \zeta(v)$ such that $\tilde{w}_u^2 = \sum_{t \in U_u} {w}^2_t $. Moreover, it is easy to see that the set $U = \bigcup_{u \in \zeta(v)} U_u$ is a minimal cover for the subtree rooted at $v$. Hence, we can write:
	\begin{align*}
\tilde{w}_v = \sqrt{\sum_{u \in \zeta(v)} \tilde{w}^2_u}
= \sqrt{\sum_{u \in \zeta(v)} \sum_{t \in U_u} {w}^2_t}
= \sqrt{\sum_{t \in U} {w}^2_t}
	\end{align*}

The property (4) directly follows the Algorithm~\ref{alg:balanced-packing-tree}.  The base case for $u \in V_C$ always holds. Conasider any non-leaf node $u$. By induction, assume $l_x = \sqrt{\sum_{v \in T_x \cap V_C} l^2_v}$ for each node $x \in \zeta(u)$. Implied by line 9 in Algorithm~\ref{alg:balanced-packing-tree}, we have
\[l_u = \sqrt{\sum_{x \in \zeta(u)} l^2_x} = \sqrt{\sum_{x \in \zeta(u)} \sum_{v \in T_x \cap V_C} l^2_v}  =  \sqrt{\sum_{v \in T_u \cap V_C} l^2_v}.\]
This concludes our proof.	
\end{proof}

It still remains to specify the position in the grid for each square assigned to a compute node.

%\begin{theorem} \label{thm:tree:cp:nopt}
%	On a symmetric tree topology, Algorithm~\ref{alg:balanced-packing-tree} computes the cartesian product $R \times S$ only in a single round, away from the optimal solution by a factor of $O(|V|)$.
%\end{theorem}
%
%
%However, the algorithm above is not strictly optimal. We use the following example to illustrate where the gap occurs. Assume $G$ is a complete ternary tree, i.e., each vertex has three children.  All edges in $G$ have the same capacity cost $w$ and data is uniformly distributed across all leaf nodes. The optimal capacity cost would be $O(\frac{N}{\sqrt{|V|}})$. If applying {\sc BalancedPacking} on a non-leaf vertex $u$, the $\square_u$ will be assigned to one of its three children. Then one single compute node will receive all input tuples by Algorithm~\ref{alg:balanced-packing-tree}, with capacity cost $O(N)$ away from optimal by a factor of $O(\sqrt{|V|})$. In the worst case, this gap can be as large as $O(|V|)$.

\paragraph{Packing squares} 
%Next we show an improved packing strategy to close this gap. Assume the $l_v$'s have been computed for all nodes in $G_\alpha$, which can be done by invoking {\sc BalancedPacking} without executing line 11. 
In this part, we discuss how we can pack each square of dimension $d_v$ assigned to leaf node $v$ inside $\square$. 
Our goal is to find an assignment (packing) of each square to compute nodes $V_C$ such that for each vertex $u$, the number of elements that cross the link $(u,p_u)$ is bounded by $O(N \cdot l_u)$. 

%\xiao{The high-level idea is that for any vertex $u$, as long as all squares associated with leaf nodes inside the subtree rooted at $u$ are packed in forms of $\{(2^i, c_i): c_i \in \{0,1,2,3\}\}$, the size of input elements that cross $(u, p_u)$ can be bounded. }
%
%Assume $N$ is an integer as the power of 2. For each compute node $v \in V_C$, we round $N \cdot l_v$ to $2^i$ if $2^i < N \cdot l_v \le 2^{i+1}$ for some integer $i$, which only increases the capacity cost by a factor of $2$.  
%
We visit all vertices in bottom-up way, starting from the leaves. We recursively assign to each node $v$ a set of squares in the form of $S_v = \{(2^i, c_i): c_i \in \{0,1,2,3\}\}$, meaning that there are $c_i$ squares of dimensions $2^i \times 2^i$.

For every leaf node $v \in V_C$, only one square is assigned to $v$ by Algorithm~\ref{alg:balanced-packing-tree}. 
Consider some non-leaf node $u$. Each of its children $v \in \zeta(u)$ is assigned with a set of squares $S_v$. We start the following procedure in an increasing order of $i \geq 0$: for each $i$, if there are $4$ squares of size $2^i \times 2^i$ in $\bigcup_{v \in \zeta(u)} S_v$, we pack them into a larger square of size $2^{i+1} \times 2^{i+1}$. In this way, we can transform $\bigcup_{v \in \zeta(u)} S_v$ into a new set of squares $S_u$, where for every $i$, there are at most $c_i \le 3$ squares of dimensions $2^i \times 2^i$. %% It is now easy show inductively that all squares of size $\leq 2^i$ can be packed inside a square of size $2^{i+1}$. Hence, we can pack all squares in $S_v'$ inside a square of size $2^{i^*+1}$, where $2^{i^*}$ is the dimension of the largest square in $S_v'$.

%%Initially, each leaf node $v \in V_C$ has been well packed with $c_j =1$ only for $j \neq \log_2 (N \cdot l_v)$.  Consider any internal node $u \in G_\alpha$ whose children are all well packed. How to pack these squares in $\bigcup_{v \in \zeta(u)} S_v$? 

Next we bound the number of elements that cross the link $(u, p_u)$ for each node $u \in V$, which is assigned with the set of squares $S_u$. Let $T_u$ be the subtree rooted $u$. Let $i^*$ be the largest integer such that $c_{i^*} \neq 0$. Note that each square of dimensions $2^i \times 2^i$ includes $2^i$ elements from both $R$ and $S$. Then, the total number of elements for all squares in $S_u$ is $\sum_i c_i \cdot 2 \cdot 2^i \le  2 \cdot (c_{i^*} +1) \cdot 2^{i^*} \le 8 \cdot 2^{i^*}$, which can be further bounded by
\begin{align*} 
%%\sum_i c_i \cdot 2 \cdot 2^i \le  & 2 \cdot (c_{i^*} +1) \cdot 2^{i^*} \\
\le & \ 8 \cdot \sqrt{\sum_{v \in T_u \cap V_C} d^2_v} \le 16 \cdot N \cdot \sqrt{\sum_{v \in T_u \cap V_C} l^2_v}  = 16 \cdot N \cdot l_u 
\end{align*}
The second inequality is implied by Algorithm~\ref{alg:balanced-packing-tree}, while the third inequality comes from the fact that  $d_v \le 2 N \cdot l_v$ for each compute node $v \in V_C$. The last equality is implied by Lemma~\ref{lem:properties}.

%%Consider any set of squares $\{(2^i,c_i): i \in \{0,1,\cdots, \log_2 N\}, c_i \in \{0,1,2,3\}\}$. For any packing solution, define its cost as the length of $x$-projection and $y$-projection. An example is shown in Figure~\ref{fig:packing-cost}.
%%\begin{figure}
%%	\centering
%%	\includegraphics[scale=0.7]{packing-cost}
%%	\caption{An illustration of packing cost. The example solution has cost with $x_1 + x_2 + y_1+ y_2$.}
%%	\label{fig:packing-cost}
%%\end{figure}
%%
%%Let $opt$ be the optimal cost of packing. For $S_v$, denote the largest index $i$ with $c_i > 0$ as $i^*$. An arbitrary packing solution has its cost bounded by $\sum_{j=1}^{i^*} c_j \cdot 2^j \le 3 \cdot \sum_{j=1}^{i^*} 2^j \le 6 \cdot 2^{i^*} \le 6 \cdot opt$. Apply this argument to each node $u \in G_\alpha$, and we can show that this algorithm has generated a packing solution away from optimal within a constant factor.

\begin{theorem} \label{thm:tree:cp:optimal}
	%On a symmetric tree topology, there is a deterministic algorithm computing cartesian product $R \times S$ optimally in one round.
	\xiao{On a symmetric tree topology $G=(V,E)$, the cartesian product $R \times S$ for $|R| = |S| = N/2$ can be computed deterministically in a single round optimally. }
\end{theorem}

\begin{proof}
To prove the correctness of the algorithm, we need to show that the packing of the squares fully covers the $|R| \times |S|$ grid. Indeed, consider the largest square $2^{i^*} \times 2^{i*}$ that occurs in the set of squares $S_r$ assigned to the root node. Observe first that we can pack all squares in $S_r$ inside a $2^{i^*+1} \times 2^{i^*+1}$ square, and thus
\[2^{2(i^*+1)} \geq \sum_v d_v^2 \geq N^2  \sum_{v \in V_C} l_v^2 = N^2\]
Hence, $2^{2i^*} \geq (N/2)\cdot(N/2) = |R| \cdot |S|$, which means that the grid is fully packed by the largest square in $S_r$.

We next show that the cost is asymptotically close to the lower bounds in Theorem~\ref{thm:lb-cp1} and Theorem~\ref{thm:lb-cp2}.	
 It can be easily checked that the number of elements transmitted through any link $e$ at the first step is at most $O\left(\min \{\sum_{v \in V_e^-} N_v, \sum_{v \in V_e^+} N_v \}\right)$, matching the lower bound in Theorem~\ref{thm:lb-cp1}.
For the second step, we have bounded the number of elements that cross link $(u,p_u)$ by $O(N \cdot l_v)$. Lemma~\ref{lem:properties} implies that $N \cdot l_v \leq N \cdot w_v/ \sqrt{\sum_{u \in U} w_u^2}$ for some minimal cover $U$ of $G^\dagger$, hence matching the lower bound in  	Theorem~\ref{thm:lb-cp2}.
\end{proof}

\subsection{Discussion on Unequal Case}
\label{sec:discussion}

At last, we discuss the difficulty of computing the cartesian product $R \times S$ with $|R| \neq |S|$ on a symmetric star topology. W.l.o.g., assume $|R| < |S|$. %%Unfortunately, the lower bound does not have an explicitly form. 
The first lower bound following the same arguement in Theorem~\ref{thm:lb-cp1} is $\Omega(C_{LB})$ where \[C_{LB} = \max_{v \in V_C} \frac{1}{w_v} \cdot \min \big \{N_v, N - N_v, |R| \big \}\] 

We next see how the counting argument yields the second lower bound under the condition $\max_v N_v < \frac{N}{2}$. Let $C$ be the cost of any correct algorithm.
Let $R_u', S_u'$ be the elements of $R,S$ received by $u$.  Then, the size of the results that can be produced at $u$ is $|R_u \cup R'_u| \cdot |S_u \cup S'_u|$. Observe the following:
\begin{itemize}
	\item $|R'_u| \leq \copt  \cdot w_u$ and $|S'_u| \leq \copt  \cdot w_u$;
	\item If $N_u < |R|$, $|R_u \cup R'_u| \leq 2 \copt  \cdot w_u$ and $|S_u \cup S_u'| \leq 2 \copt  \cdot w_u$ 
	\item If $N_u \ge |R|$, $\copt  \cdot w_u \ge |R|$.
\end{itemize}
Summing over all node,  the total size of the output must be at least $|R| \cdot |S|$. We then obtain  
\begin{align*}
|R| \cdot |S| \le & \sum_{u \in V_C} (|R'_u| + |R''_{u}|) (|S'_u| + |S''_u|) \\
	\leq  & \sum_{u \in V_C: N_u < |R|} 2 \min\{C\cdot w_v, |R|\} \cdot 2 \min\{C \cdot w_v, |S|\} 
	+  \sum_{u \in V_C: N_u \ge |R|} |R| \cdot  \{C \cdot w_v + S_u, |S|\}
\end{align*}
whose minimizer gives the second lower bound, which becomes rather complicated without a clean form as Theorem~\ref{thm:lb-cp2}.

This is just an intuition of why the unequal case would make the lower bound hard even on the symmetric star. In Appendix~\ref{appendix:cp-unequal}, we give a more detailed analysis on the lower bound, as well as an optimal algorithm. Extending our current result to the general symmetric tree topology is left as future work.

	\section{Sorting}
\label{sec:sorting}

In the sorting problem, we are given a set $R$ whose elements are drawn from a totally ordered domain. %%Moreover, elements with the same values are accepted, but their ties will not be considered in the distributed fashion. 
We first define an ordering of compute nodes in the following way: after picking an arbitrary node as the root, any left-to-right traversal of the underlying network tree is a valid ordering of compute nodes.  The goal is to redistribute the elements of $R$ such that on an ordering of compute nodes as $v_1, v_2,\cdots,v_{|V_C|}$,  elements on node $v_i$ are always smaller than those on node $v_j$ if $i < j$. 

Given an initial distribution $\dstr$ of the data across the compute nodes, we denote by %$R^\dstr_v$ the elements in node $v$. Let 
$N^\dstr_v$ the initial data size in node $v$. %We assume $N^\dstr_v > 0$ for each $v \in V_C$; otherwise, we remove such a nodes from the topology. 
Whenever the context is clear, we drop the superscript $\dstr$ from the notation.

\subsection{Lower Bound}
\label{sec:sorting-lb}

Our lower bound for sorting has the same form as the one for set intersection, with the only difference that the cost is expressed as tuples, and not bits. %We defer the full proof to the appendix.

\begin{theorem} \label{thm:lb:sorting}
	Let $G=(V,E)$ be a symmetric tree topology.
	Any algorithm sorting elements in a set $R$ has (tuple) cost $\Omega(C_{LB})$, where
	\[C_{LB} = \max_{e \in E} \frac{1}{w_e} \cdot \min \left \{\sum_{v \in V^-_e} N_v,  \sum_{v \in V^+_e} N_v \right\}. \]
\end{theorem}

\begin{figure}
	\centering
	\includegraphics[scale=1.2]{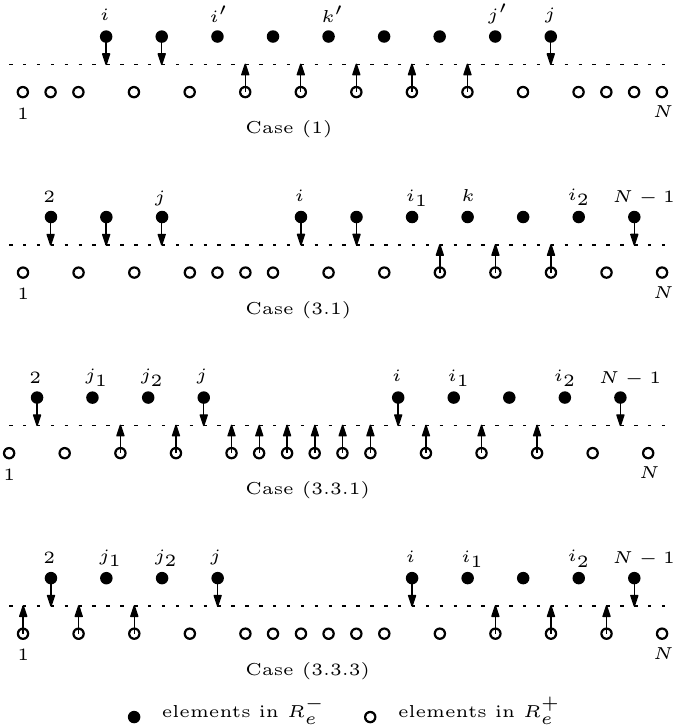}
	\caption{Data exchange between $V^-_e, V^+_e$.}
	\label{fig:case}
\end{figure}

%%We only give some intuition here and the whole proof can be found in Appendix~\ref{appendix:lb:sorting}.
\begin{proof}
	We construct an initial data distribution as follows. Assume elements in $R$ are ordered as $r_1, r_2, \cdots, r_N$, where $i$ is the rank of element $r_i$ in $R$. Without loss of generality, assume $N$ is even. We assign elements to compute nodes in the ordering of $\{r_1, r_3, \cdots, r_{N-1}, r_2, r_4, \cdots, r_{N}\}$. Moreover, we pick one arbitrary route node of $G$ as the root, where all compute nodes are leaves of the tree. All compute nodes in $V_C$ are also labeled as $v_1, v_2,\cdots, v_{|V_C|}$ in an left-to-right traversal ordering, i.e., recursively traversing the leaves in the left subtree and then the right subtree. For example, the node $v_1$ with initial data size $N_1$ will be assigned with elements $\{r_1, r_3, \cdots, r_{2N_1-1}\}$ if $N_1 \le \frac{N}{2}$, and $\{r_1, r_3, \cdots, r_{N-1}, r_2, r_4,\cdots, r_{2N_1 - N}\}$ otherwise. We need to argue that any algorithm correctly sorting $R$ under this initial distribution must have a cost $\Omega(C_{LB})$. 
	
	Consider an arbitrary edge $e \in E$. Removing $e$ defines a partition of $V_C$ as $V^-_e, V^+_e$. Denote $R^-_e = \bigcup_{v \in V^-_e} R_v$ and $R^+_e = \bigcup_{v \in V^+_e} R_v$. It should be noted that $R^-_e$ or $R^+_e$ is a sub-interval of $\{r_1, r_3, \cdots,r_{N-1}, r_2, r_4,\\ \cdots, r_N\}$, or a  sub-interval of $\{r_2, r_4, \cdots, r_N, r_1, r_3, \cdots, r_{N_1}\}$. Note that every element transmitted between $V^-_e$ and $V^+_e$ must go through edge $e$.  Without loss of generality, assume $|R^-_e| \le \frac{N}{2} \le |R^+_e|$. Then it suffices to show that the total number of elements exchanged between $V^-_e$ and $V^+_e$ is at least $\Omega(|R^-_e|)$. %%For example, if $R^-_e$ is a sub-interval of $\{r_1, r_3, \cdots,r_{N-1}\}$ and three elements $r_i < r_k < r_j$ are not sent from $R^-_e$ to $R^+e$, then all even-rank elements in $[r_i, r_j]$\footnote{We use $[a,b]$ to denote $\{a,a+1, \cdots,b\}$ for short. } must be sent from $R^+_e$ to $R^-_e$, implied by the correctness of sorting algorithms. The proof in Appendix~\ref{appendix:lb:sorting} is a case-by-case analysis, but all are based on such a similar argument.

	In the extreme case, there is only one element in $R^-_e$, say $R^-_e = \{r_i\}$. If $r_i$ is not sent through $e$, at least one element in $R^+_e$ must be sent to $V^-$; otherwise, no comparison between $r_i$ and any element $r_j \in R^+_e$ is performed, contradicting to the correctness of algorithms. So at least one element is transmitted through edge $e$. In general, at least two elements are in $R^-_e$. We further distinguish four cases: (1) $r_2 \notin R^-_e, r_{N} \notin R^-_e$; (2) $r_1 \notin R^-_e, r_{N-1} \notin R^-_e$; (3) $r_2 \in R^-_e, r_{N-1} \in R^-_e$; (4) $r_1 \in R^-_e, r_{N} \in R^-_e$. Note that (2) can be argued symmetrically with (1) and (4) can be argued symmetrically with (3). 
	
	Case (1): $r_2 \notin R^-_e, r_{N} \notin R^-_e$. In this case, the $R^-_e$ must be a subset of $\{r_1, r_{3}, \cdots, r_{N-1}\}$. Let $i, j$ be the smallest and largest rank of elements in $R^-_e$. If all elements in $R^-_e$ have been sent from $V^-_e$ to $V^+_e$, then we are done. Otherwise, let $i', j'$ be the smallest and largest rank of elements in $R^-_e$ which are not sent from $V^-_e$ to $V^+_e$. Furthermore, if all elements in $R^-_e - \{r_{i'}, r_{j'}\}$ are sent from $V^-_e$ to $V^+_e$, it can be easily checked that the number of such elements is at least $\frac{|R^-_e|}{2}$. Otherwise, there is $r_{k'} \in R^-_e - \{r_{i'}, r_{j'}\}$ not sent from $V^-_e$ to $V^+_e$. By the definition, $r_{i'} < r_{k'} < r_{j'}$. Implied by the ordering of compute nodes, all elements in $[r_{i'}, r_{j'}]$ should reside on $V^-_e$ when the algorithm terminates. In this case, each element in $[r_{i'}, r_{j'}] - R^-_e$ should be sent from $V^+_e$ to $V^-_e$, and each in $\{r_i, r_{i+2}, \cdots, r_{i'-2}\} \cup \{r_{j'+2}, r_{j'+4}, \cdots, r_j\}$ are sent from $V^-_e$ to $V^+_e$, as illustrated in Figure~\ref{fig:case} (Due to the page limit, the figure is moved to Appendix~\ref{appendix:lb:sorting}). So the number of elements transmitted through edge $e $ is at least $\frac{i'-i}{2} + \frac{j-j'}{2} + \frac{j'-i'}{2} = \frac{j - i}{2} \ge |R^-_e| -1 \ge \frac{|R^-_e|}{2}$. 
	
	Case (3): $r_{N-1} \in R^-_e, r_2 \in R^-_e$. Let $i$ be the smallest odd rank and $j$ be the largest even rank of elements in $R^-_e$. Note that $j < i$ since $|R^-_e|  \le \frac{N}{2}$. We further consider three cases as below. 
	
	Case (3.1): all elements in $\{r_2, r_4, \cdots, r_j\}$ are sent from $V^-_e$ to $V^+_e$. If all elements in $\{r_i,r_{i+2}, \cdots,\\r_{N-1}\}$ are also sent from $V^-_e$ to $V^+_e$, then we are done. Otherwise, let $i_1, i_2$ be the smallest and largest rank of elements in $\{r_i,r_{i+2}, \cdots,r_{N-1}\}$ not sent from $V^-_e$ to $V^+_e$. Furthermore, if all elements in $\{r_i,r_{i+2}, \cdots,r_{N-1}\} - \{r_{i_1}, r_{i_2}\}$ are sent from $V^-_e$ to $V^+_e$, it can be easily checked that the number of elements sent from $V^-_e$ to $V^+_e$ is at least $\frac{|R^-_v|}{2}$. Otherwise, there is $r_{k'} \in \{r_i,r_{i+2}, \cdots,r_{N-1}\} - \{r_{i_1}, r_{i_2}\}$ not sent from $V^-_e$ to $V^+_e$. By the definition, $r_{i_1} < r_{k} < r_{i_2}$. Implied by the ordering of compute nodes, all elements in $[r_{i_1}, r_{i_2}]$ should reside on $V^-_e$ when the algorithm terminates. In this case, each element in $[r_{i_1}, r_{i_2}] - R^-_e$ should be sent from $V^+_e$ to $V^-_e$, and each in $\{r_2, r_4, \cdots, r_j\} \cup \{r_i, r_{i+2}, \cdots, r_{i_1-2}\} \cup \{r_{i_2+2}, r_{i_2+4}, \cdots, r_{N-1}\}$ are sent from $V^-_e$ to $V^+_e$, as illustrated in Figure~\ref{fig:case}. So the number of elements transmitted through edge $e $ is at least $\frac{j}{2} + \frac{i_1 - i}{2} + \frac{N-1-i_2}{2} + \frac{i_2 - i_1}{2} = \frac{N-1+j-i}{2}  = |R^-_e| -1 \ge \frac{|R^-_e|}{2}$. 
	
	Case (3.2): all elements in $\{r_i,r_{i+2}, \cdots,r_{N-1}\}$ are sent from $V^-_e$ to $V^+_e$, which can be argued symmetrically.
	
	Case (3.3): at least one element in $\{r_2, r_4, \cdots, r_j\}$ and one element in $\{r_i,r_{i+2}, \cdots,r_{N-1}\}$ are not sent from $V^-_e$ to $V^+_e$. Let $j_1, j_2$ be the smallest and largest even rank of elements in $R^-_e$ not sent from $V^-_e$ to $V^+_e$. Let $i_1, i_2$ be the smallest and largest odd rank of elements in $R^-_e$ not sent from $V^-_e$ to $V^+_e$. Note that each element in $\{r_2, r_4, \cdots, r_{j_1 -2}\} \cup \{r_{j_2 + 2}, r_{j_2+4}, \cdots, r_j\} \cup \{r_i, r_{i+2}, \cdots, r_{i_2-2}\} \cup\{r_{i_2 + 2}, r_{i_2 + 4}, \cdots, r_{N-1}\}$ is sent from $V^-_e$ to $V^+_e$.  
	
	By the ordering of compute nodes, (3.3.1) all elements in $[r_{j_1}, r_{i_2}]$ or (3.3.2) all elements in $[r_1, r_{j_2}] \cup [r_{i_1}, r_{N}]$ should reside on $V^-_e$ when the algorithm terminates. In (3.3.1), each element in $[r_{j_1}, r_{i_2}] - R^-_e$ should be sent from $V^+_e$ to $V^-_e$, as illustrated in Figure~\ref{fig:case}. The number of elements transmitted through edge $e$ is at least $i_2 - j_1 + 1 - \frac{j-j_1}{2} - \frac{i_2 - i}{2} + \frac{j_1 -2}{2} + \frac{j-j_2}{2} + \frac{i_1 -i}{2} + \frac{N-1 - i_2}{2} \ge \frac{N-1 + i_1 - j_2}{2} \ge \frac{N}{2} \ge |R^-_e|$. In (3.3.2), each element in $\{r_1,r_3,\cdots,r_{j_2-1}\} \cup \{r_{i_1+1},r_{i_1+3}, \cdots, r_{N-1}\}$ should be sent from $V^+_e$ to $V^-_e$, as illustrated in Figure~\ref{fig:case}. The number of elements transmitted through edge $e$ is at least
	$\frac{j_2+1}{2} +  \frac{N - i_1+1}{2} + \frac{j_1 -2}{2} + \frac{j-j_2}{2} + \frac{i_1 -i}{2} + \frac{N-1 - i_2}{2}= \frac{N-1-i_2+j_1}{2} \ge \frac{N}{2} \ge |R^-_e|$. 
\end{proof}

\subsection{A Sampling-based Algorithm}

In the MPC model, the theoretically optimal sorting algorithm inherited from~\cite{goodrich1999communication} is rather complicated. Instead, sampling-based techniques, such as TeraSort~\cite{o2008terabyte}, are more amenable to be extended to more complex networks.  In this section, we present a randomized communication protocol for a symmetric tree topology, named {\em weighted TeraSort} (wTS), which generalizes the TeraSort algorithm in three fundamental ways. {First, TeraSort is designed for the MapReduce~\cite{dean2004mapreduce} framework, which is an instantiation of the theoretical MPC model (with star topology), and we extend it to the general tree topology.}
Second, not all nodes participate in the splitting of the data, but only the ones that initially have a substantial amount of data. Third, we do not split the data uniformly, but proportionally to the size of the initial data. Before introducing our algorithm, we revisit the TeraSort algorithm.

{
\paragraph{TeraSort Algorithm} It picks an arbitrary node as the coordinator.  Set $\rho =4 \cdot  \frac{|V_C|}{N} \ln(|V_C| \cdot N)$. 
\begin{description}
\item[Round 1:] Each node $u \in V_C$ samples each element from its local storage with uniform probability $\rho$, and sends all sampled elements to the coordinator.  Let $s$ be the number of samples generated in total.
\item[Round 2:]The coordinator sorts all sampled elements received. Let $b_i$ be the $i \cdot \lceil \frac{s}{|V_C|}\rceil$-th smallest object in the sorted samples for $i \in \{1,2,\cdots, |V_C|-1\}$, $b_0 = -\infty$ and $b_{|V_C|} = +\infty$. It then broadcasts $|V_C| + 1$ splitters  $b_0, b_1, \cdots, b_{|V_C|}$ to all nodes.
\item[Round 3:] Upon receiving all splitters, each node scans it own elements. For each element $x$, the node finds the two consecutive splitters $b_i$ and $b_{i+1}$ such that $b_i \le x < b_{i+1}$ and then sends $x$ to $v_{i+1}$. Finally, each node locally sorts all elements that it has received.
\end{description}
}

%\begin{algorithm}[t]
%	\caption{{\sc ReDistribute}($G, \dstr$)}
%	\label{alg:redistribute}
%	
%	$V_H \gets \{u \in V_C: N_v \ge |V_C|\}, V_L \gets V_C - V_H$\; 
%	Label $V_H$ as $\{v_1, v_2, \cdots, v_k\}$\;
%	\ForEach{$u \in V_L$}{
%		$\{N^1_u,  N^2_u, \cdots, N^k_u\} \gets \textrm{{\sc Proportional}}(V_H, u)$\;	
%		Send $N^i_u$ local elements (without duplication) to node $v_i$\; 
%	} 
%\end{algorithm}

\begin{algorithm}[t]
	\caption{{\sc Proportional}($V_H, u$)}
	\label{alg:proportional-distribute}
	
	$\Delta \gets 0$, $i \gets 1$\;
	\While{$i \le k$}{
		$x \gets \frac{N_{v_i}}{\sum_{v \in V_H} N_v} \cdot N_u$\;
		\If{$\Delta \ge x-\lfloor x \rfloor$}{
			$N^i_u \gets \lfloor x \rfloor, \quad \Delta \gets \Delta - (x-\lfloor x \rfloor)$\;
		}
		\Else{
			$N^i_u \gets \lfloor x \rfloor+1, \quad \Delta \gets \Delta +1 -(x-\lfloor x \rfloor)$\;
		}
		$i \gets i+1$;
	}
	\Return $N^1_u,  N^2_u, \cdots, N^k_u$\;
\end{algorithm}

{Now we describe our algorithm.}
Assume that the data statistics $N_v$'s are known to all compute nodes. A compute node $v \in V_C$ is {\em heavy} if $N_v \ge |V_C|$ and {\em light} otherwise. Let $V_H, V_L \subseteq V_C$ be the set of heavy and light compute nodes respectively. For simplicity, we pick an arbitrary non-compute node as the root and label heavy nodes in $V_H$ from left to right as $v_1, v_2, \cdots, v_{k}$.
\begin{description}
	\item[Round 1:] %Redistribute the initial data by invoking Algorithm~\ref{alg:redistribute} if $V_L \neq \emptyset$.
	Each light node $u \in V_L$ sends its local data to heavy nodes proportional to $N_{v_i}$'s. More specifically, node $u$ sends $N^i_u$ local elements to $v_i$ for each $i \in \{1,2,\cdots,k\}$, where $N^i_u$ is computed by Algorithm~\ref{alg:proportional-distribute}. 
	Let $M_j$ be the number of elements residing on heavy node $v_j$ after round 1.

	\item[Round 2:]  Each heavy node samples each element from its local storage with the same uniform probability $\rho$ % = 4 \cdot \frac{|V_C|}{N} \ln(|V_C| \cdot N)$ 
	independently and then sends the sampled elements to $v_1$. Let $s$ be the number of samples generated in total.
	
	\item[Round 3:] Node $v_1$ sorts all samples received.  %%and picks $|V_C|+1$ splitters $b_0,b_1, b_2, \cdots, b_{|V_C|-1}, b_{|V_C|}$ where $b_i$ is the $i \cdot \lceil \frac{s}{|V_C|} \rceil $-th smallest element among all samples, $b_0 = -\infty$, and $b_{|V_C|}= +\infty$. Then, $v_1$ broadcasts these splitters to the remaining heavy compute nodes. 
	Let $t_i$ be the $i \cdot \lceil \frac{s}{|V_C|} \rceil $-th smallest element among all samples. Let $c_j = \lceil \frac{|V_C|}{N} \cdot M_j \rceil$. It chooses $k+1$ splitters as follows: (1) $b_0 = -\infty$; (2) $b_i = t_j$ where $j = c_1 + c_2 +\cdots +c_i$; (3) $b_{k}= +\infty$. Then, $v_1$ broadcasts $b_0, b_1, \cdots, b_k$ to the remaining heavy compute nodes. 
	
	%%\item[Round 4:] Upon receiving the splitters, each heavy node $v_j$ scans its own elements. For each $i \in \{0,1,\cdots, |V_C|-1\}$, it computes the number of local elements falling into the interval $[b_i, b_{i+1})$, denoted by $c[i][j]$, and sends $c[i][j]$ to $v_i$. 
	%%
	%%\item[Round 5:] Each node $v_i$ computes the quantity $c_i = \sum_{j=1}^k c[i][j]$ and sends $c_i$ to $v_1$. 
	%%
	%%\item[Round 6:] Upon receiving $c_1, c_2, \cdots, c_k$, node $v_1$ merges the intervals defined by splitters $b_1, b_2, \cdots, b_k$ if necessary. More specifically, it computes a subset of splitters $B = b'_0, b'_1, \cdots, b'_{\ell}$ by invoking Algorithm~\ref{alg:merge-splitter}. After that, it sends $B$ to all heavy nodes.
	%%
	\item[Round 4:] Upon receiving all splitters, each heavy node scans its own elements. For each element $x$, the node finds the two consecutive splitters $b_i$ and $b_{i+1}$ such that $b_i \le x < b_{i+1}$ and then sends $x$ to $v_{i+1}$. Finally, each node locally sorts all elements that it has received.
	\end{description}

%The proof of Theorem~\ref{thm:wst} is given in Appendix~\ref{appendix:wst}. 
A  possible improvement is that if the maximum data that some node holds exceeds $N/2$, every node just sends their data to that node. Otherwise, we simply run the wTS routine on the whole topology.

 \subsection{Analysis}
 Before proving the complexity for our algorithm,  we first point out some important properties.

 \begin{lemma}
 	\label{lem:light-distribute}
 	Consider a light node $u \in V_L$. Then, the following hold true in Algorithm~\ref{alg:proportional-distribute}:
 	\begin{enumerate}
 		\item  for any $i \in [k]$, $\sum_{j=1}^i N^j_u -1 \le  \frac{\sum_{j=1}^i N_{v_j}}{\sum_{j=1}^k N_{v_j}} \cdot N_u \le \sum_{j=1}^i N^j_u$;
 		\item for any $i_1, i_2 \in [k]$ with $i_1 < i_2$, $\sum_{j=i_1}^{i_2} N^j_u  \leq \frac{\sum_{j=i_1}^{i_2} N_{v_j}}{\sum_{j=1}^k N_{v_j}} \cdot N_u  + 1$. 
 		\item $\sum_{j=1}^k N^j_u \ge N_u$.
 	\end{enumerate}
 \end{lemma}
 
 \begin{proof}
 	We first prove (1) by induction. The base case $i =1$ follows since $N^1_u = \lfloor\frac{N_{v_1}}{\sum_{j=1}^k N_{v_j}} \cdot N_u \rfloor + 1$. For the inductive step, assume the claim holds for $i$. Let $\Delta_i$ be the value of $\Delta$ after being updated during the $i$-th iteration of the while loop. Observe that the invariant $\Delta_i = \sum_{j=1}^i N^j_u - \frac{\sum_{j=1}^i N_{v_j}}{\sum_{j=1}^k N_{v_j}} \cdot N_u$ always holds. It can also be checked that $\Delta_i \ge 0$ since $0 \le x-\lfloor x \rfloor \le 1$. 
 	
 	Consider the $(i+1)$-th iteration of while loop. When it goes into line 4, we have: 
 	\[\sum_{j=1}^{i+1} N^j_u = N^{i+1}_{u} + \sum_{j=1}^i N^j_u =  \left \lfloor \frac{N_{v_{i+1}}}{\sum_{j=1}^k N_{v_j}} \cdot N_u  \right \rfloor + \Delta_i + \frac{\sum_{j=1}^i N_{v_j}}{\sum_{j=1}^k N_{v_j}} \cdot N_u = (\Delta_i -x + \lfloor x \rfloor ) + \frac{\sum_{j=1}^{i+1} N_{v_j}}{\sum_{j=1}^k N_{v_j}} \cdot N_u\]
 	In this case, $0< \Delta_i - x+ \lfloor x \rfloor < 1$, so the claim holds. When the algorithm goes into line 6, 
 	\[ \sum_{j=1}^{i+1} N^j_u = N^{i+1}_{u} + \sum_{j=1}^i N^j_u = \left \lfloor \frac{N_{v_{i+1}}}{\sum_{j=1}^k N_{v_j}} \cdot N_u  \right \rfloor + 1 + \Delta_i + \frac{\sum_{j=1}^i N_{v_j}}{\sum_{j=1}^k N_{v_j}} \cdot N_u = (\Delta_i + 1 - x + \lfloor x \rfloor) + \frac{\sum_{j=1}^{i+1} N_{v_j}}{\sum_{j=1}^k N_{v_j}} \cdot N_u \]
 	
 	\smallskip
 	
 	We prove (2) based on (1). Observe that 
 	\[ \sum_{j=i_1}^{i_2} N^j_u = \sum_{j=1}^{i_2} N^j_u - \sum_{j=1}^{i_1} N^j_u \le \frac{\sum_{j=1}^{i_2} N_{v_j}}{\sum_{j=1}^k N_{v_j}} \cdot N_u + 1  - \frac{\sum_{j=1}^{i_1} N_{v_j}}{\sum_{j=1}^k N_{v_j}} \cdot N_u \le \frac{\sum_{j=i_1}^{i_2} N_{v_j}}{\sum_{j=1}^k N_{v_j}} \cdot N_u  + 1 \]
 	We can obtain a similar expression for $i_2$; then the claim holds by adding the two inequalities.
 	
 	\smallskip
 	Property (3) follows immediately from (1) by setting $i=k$.
 \end{proof}

\begin{theorem} \label{thm:wst}
	Let $G = (V, E)$ be a symmetric tree topology and $R$ be an ordered set of $N$ elements. If $N \ge 4 |V_C|^2 \cdot \ln (|V_C| \cdot N)$, with probability $1 - \frac{1}{N}$,  the wST algorithm sorts $R$ in 4 rounds with cost $O(1)$ away from the optimal.
\end{theorem}

\begin{proof}
	From property (3) of Lemma~\ref{lem:light-distribute}, it follows that all the data of the light nodes is sent to the heavy nodes during the first round. Hence, the algorithm will produce the correct sorting. We complete the proof of Theorem~\ref{thm:wst} by analyzing the cost of the wTS algorithm.
	
	First, we observe that at least half the data is distributed across heavy nodes initially, i.e., $\sum_{j =1}^k N_{v_j} \ge \frac{N}{2}$. Indeed, the size of initial data distributed across all light node is strictly smaller than $\frac{N}{2|V_C|} \cdot |V_C| = \frac{N}{2}$, so the remaining data with size at least $\frac{N}{2}$ must reside on heavy nodes. We next analyze the cost for each round separately.
	
	\medskip \noindent \introparagraph{Round 1}
	Consider an arbitrary edge $e \in E$, which defines a partition of compute nodes $V^-_e, V^+_e$. If $V_H \cap V^+_e \neq \emptyset$, it holds that $V_H \cap V^+_e = \{v_i, v_{i+1}, \cdots, v_j\}$ or $\{v_1, v_2, \cdots, v_i\} \cup \{v_j, v_{j+1}, \cdots, v_k\}$ for some $i,j \in [k]$ and $i \le j$. For any light node $u \in V_L$, the number of data sent to the nodes in $V_H \cap V^+_e$ can then be bounded as follows using Lemma~\ref{lem:light-distribute}(2):
	\[\sum_{v \in V^+_e \cap V_H} N^v_u \le 2 + \sum_{v \in V^+_e \cap V_H}   \frac{N_{v}}{\sum_{v' \in V_H} N_{v'}} \cdot N_u\] 
	In this way, the number of data sent from light nodes in $V^-_e$ to heavy nodes in $V^+_e$ can be bounded as 
	\begin{align*}
	& \sum_{u \in V^-_e \cap V_L} \left(2 + \sum_{v \in V^+_e \cap V_H}   \frac{N_{v}}{\sum_{v' \in 	V_H} N_{v'}} \cdot N_u \right) 
	\le \sum_{u \in V^-_e \cap V_L} 2 + \sum_{u \in V^-_e \cap V_L} \sum_{v \in V^+_e \cap V_H} \frac{2N_{v}}{N} \cdot N_u \\
	& \le 2 \min\left\{\sum_{u \in V^-_e} N_u, |V_C|\right\} + \frac{2}{N} \cdot \left(\sum_{u \in V^-_e} N_u \right) \cdot \left(\sum_{v \in V^+_e} N_v\right)  
	\le 4 \min\left\{\sum_{u \in V^-_e} N_u, \sum_{v \in V^+_e} N_v\right\}
	\end{align*}
	The rationale behind the third inequality is that $|V_C| \le \frac{N}{2|V_C|} \le \sum_{v \in V^+_e \cap V_H} N_v \le \sum_{v \in V^+_e} N_v$ and $\frac{a \cdot b}{a+b} \le \min\{a,b\}$ holds for any $a,b \ge 1$. If $V_H \cap V^-_e \neq \emptyset$, we can make a symmetric argument.

	We observe here that the number of data received by any heavy node $v \in V_H$ in round 1 is at most 
	\begin{align*}
	\sum_{u \in V_L} \left \lceil \frac{N_v}{\sum_{v' \in V_H} N_{v'}} \cdot N_u \right \rceil =  \sum_{u \in V_L}  \frac{N_v}{\sum_{v' \in V_H} N_{v'}} \cdot N_u +  \sum_{u \in V_L} 1
	\le  \frac{2N_v}{N} \cdot  \sum_{u \in V_L} N_u + |V_C| \le 3 N_v
	\end{align*}
	where the rationale behind the first inequality is that $\sum_{v' \in V_H} N_{v'} \ge \frac{N}{2}$ and that behind the second inequality is that $|V_C| \le \frac{N}{2|V_C|} \le N_v$. 
	Hence, for every heavy node $v$, $M_v \leq 3 N_v + N_v = 4 N_v$.	
	
	\medskip \noindent \introparagraph{Rounds 2, 3}
	During sampling, each element is an independent Bernoulli sample, so we have $E[s] = \rho N$. Applying the Chernoff bound, $\Pr[s \ge 2 \rho N] \le \exp\left(-\Omega(\rho N)\right)$. In round 2 and round 3, the number of elements received or sent by any node is at most $s$, which is smaller than $2\rho N$ with probability at least $1- \exp\left(-\Omega(\rho N)\right) \ge 1 - (\frac{1}{|V_C| \cdot N})^{4|V_C|}$. Observe that $2 \rho N \leq N / |V_C|$. Since there is a heavy node at each side of an edge that has data getting through, we have $2 \rho N \leq \min\{\sum_{u \in V^-_e} N_u, \sum_{v \in V^+_e} N_v\} $.
	
	\medskip \noindent \introparagraph{Round 4} In this round, each heavy node $v_i$ sends out at most $M_i$ elements and receives all the elements falling into the interval $[b_i, b_{i+1})$, i.e., $R \cap [b_{i-1}, b_i)$. Let $t_0 = -\infty$ and $t_{|V_C|} = +\infty$. Under the condition that $s \le 2\rho N$, we first observe that for any $j \in \{1,2,\cdots, |V_C|\}$, $|R \cap [t_{j-1}, t_j)| \le 8 \cdot \frac{N}{|V_C|}$, which holds with probability at least $1-\frac{1}{N}$, following a similar analysis to~\cite{tao2013minimal}. Together, the probability that all these assumptions hold is 
	\[ \left(1 - (\frac{1}{|V_C| \cdot N})^{4|V_C|} \right) \cdot \left(1- \frac{1}{4N}\right) \ge 1- \frac{1}{N}\]

	Conider any heavy compute node $v_j$. The number of intervals allocated to $v_j$ is exactly $c_j$, thus the number of elements recieved by $v_j$ in the last round is at most 
	\[\lceil \frac{M_j}{N} \cdot |V_C|\rceil  \cdot 8 \cdot \frac{N}{|V_C|} \le (\frac{|M_j|}{N} \cdot |V_C| +1 )  \cdot 8 \cdot \frac{N}{|V_C|} \le M_j + 8 \frac{N}{|V_C|} \le 4 N_{v_j}+ 16 N_{v_j} = O(N_{v_j}) \]
	with probability at least $1- 1/N$. 
	
	Next we bound the amount of data transmitted on every link $e \in E$. Removing $e$ will partition the compute nodes in $V^{-}_e, V^{+}_e$. W.l.o.g., assume $\sum_{v \in V^{-}_e \cap V_H} N_v \le \sum_{v \in V^{+}_e \cap V_H} N_v$. The size of data sent from the heavy nodes in $V^{-}_e$ to $V^{+}_e$ is always bounded by the total size of data sitting in $v \in V^{-}_e \cap V_H$, with $O(\sum_{v \in V^{-}_e \cap V_H} N_v) = O(\min\{\sum_{v \in V^{-}_e \cap V_H} N_v, \sum_{v \in V^{+}_e \cap V_H} N_v\})$.  The size of data sent from the heavy nodes in $V^{+}_e$ to $V^{-}_e$ is at most the number of elements recieved by all compute nodes in $V^{-}_e \cap V_H$, thus bounded by $O(\sum_{v \in {V^-_e \cap  V_H}} N_v) = O(\min\{\sum_{v \in V^{-}_e \cap V_H} N_v, \sum_{v \in V^{+}_e \cap V_H} N_v\})$. In either way, the capacity of each edge $e$ is matched by its lower bound, thus completing the proof.
\end{proof}

	\section{Related work}

The fundamental difference of the topology-aware model we use with other parallel models (e.g., BSP~\cite{BSP},  MPC~\cite{beame:communication}, LogP~\cite{LogP}) is that the cost depends both on the topology and properties of the network and the nodes. Prior models view the network as a star topology, where each link and each node have exactly the same cost functions. In this sense, our model can be viewed as a generalization, where the topology and the node heterogeneity is taken into account. 

There have already been some efforts to introduce topology-aware models, including~\cite{chattopadhyay2017tight, langberg2019topology} as mentioned in the introduction.

One line of work in distributed computing on networks are the classical LOCAL and CONGEST models~\cite{linial1992locality, peleg2000distributed}, where distributed problems are also considered in networks modeled as an arbitrary graph. These two models differentiate from ours in two important aspects. First, in each round, each node can only communicate with its neighbors; instead, in our model we can send messages to other nodes that may be located several hops away. Second, the target is to design algorithms that minimize the number of rounds. As a combination of both aspects, the diameter of the communication network cannot be avoided as a cost in these models.  Moreover, system synchronization after each round is a huge bottleneck of modern massively parallel systems; thus, any algorithm in these two models running in non-constant number of rounds would become hard to implement efficiently in practice. 

Network routing has been studied in the context of parallel algorithms
(see \cite{leighton1983complexity,leighton2014introduction}),
distributed computing (see, e.g.~\cite{leighton1994packet}), and
mobile networks \cite{madden2003design}.
Several general-purpose optimization methods for network problems have been proposed \cite{palomar2006tutorial}.
Our proposed research deviates from prior literature by considering a ``distribution-aware'' setting, and tasks that have not been
considered before.

The topology-aware model we use in this paper has been previously used to design algorithms for aggregation~\cite{LiuSBS18}. However, only star topologies were considered. Madden et al.~\cite{madden02osdi, madden02mcsa} also proposed a tiny aggregation service which does
topology-aware in-network aggregation in sensor networks.
Culhane et al.~\cite{Culhane2014hotcloud,Culhane2015infocomm}
propose LOOM, a system that builds an aggregation tree with fixed fan-in for all-to-one aggregations, and
assigns nodes to different parts of the plan according to the amount of data reduced during aggregation. 
Chowdhury et al.~\cite{orchestra11sigcomm} propose Orchestra, a system to manage network activities
in MapReduce systems. Both systems are cognizant of the network topology, but agnostic to the distribution of
the input data. They also lack any theoretical guarantees.

	\section{Conclusion}
\label{sec:conclusion}

In this paper, we studied three fundamental data processing tasks in a topology-aware massively parallel computational model. We derived lower bounds based on the cardinality of the initial data distribution at each node and we designed provably optimal algorithms for each task with respect to the initial data distribution. Interestingly, these problems have different dependency on the topology structure, the cost functions (bandwidth), as well as the data distribution. 

There are several exciting directions for future research. For one, we would like to extend our algorithms and lower bounds to non-symmetric and general (non-tree) topologies. General topologies (e.g., grid, torus) are particularly challenging because there are multiple routing paths between two compute nodes, and thus a topology-aware algorithm needs to consider all nodes in the routing path, instead of just the destination. Looking further ahead, it would be interesting to study more complex tasks that have so far been analyzed only in the context of the MPC model, starting from a simple join between two relations, and continuing to ensembles of tasks in more complex queries.

	\bibliographystyle{abbrv}
	\bibliography{other,master,bibfile}
	
	\clearpage
	\onecolumn
	\appendix

\section{Omitted Proofs}

\subsection{Cartesian product in Unequal Size}
\label{appendix:cp-unequal}
We consider the general cartesian product on a symmetric star topology $G = (V,E)$. For simplicity, we divide the compute nodes into two subsets:
\begin{align*}
	V_\alpha = \{v \in V_C: \min\{N_v, N-N_v\} < |R|\}, \ V_\beta =V_C - V_\alpha
\end{align*}
The first lower bound can be simplified as follows.
\begin{theorem}
	\label{thm:lb-cp-unequal-1}
	Any algorithm computing cartesian product $R \times S$ has cost $\Omega(C)$, where 
	\[C \ge \max \left \{\max_{v \in V_\alpha} \frac{\min\{N_v, N -N_v\}}{w_v}, \max_{v \in V_\beta} \frac{|R|}{w_v}\right \}\]
\end{theorem}

Moreover, we define $V(R,S,V_C)$ as the minimizer for the following formula.
\begin{equation}
	\label{eq:1}
	\sum_{v \in V_C} \min\{C \cdot w_v, |R|\} \cdot C \cdot w_v \ge |R| \cdot |S|
\end{equation}
Then we are able to give the second lower bound as below.

\begin{theorem}
	\label{thm:lb-cp-unequal-2}
 If $\max_v N_v \le \frac{N}{2}$, any algorithm computing cartesian product $R \times S$ has cost $\Omega(C)$, where 
 \[C \ge \min \left \{\frac{|S|}{\max_v w_v}, \frac{\sum_{u \in V_\alpha}|S_u|}{2 \sum_{u \in V_\beta} w_u}, V(R, \cup_{u \in V_\alpha} S_u, V_\alpha)\right \}\]
\end{theorem}

\begin{proof}
	It suffices to show that if $C \le |S|/\max_v w_v$, then $C \ge \min \left\{ \frac{\sum_{u \in V_\alpha}|S_u|}{2 \sum_{u \in V_\beta} w_u}, V(R, \cup_{u \in V_\alpha} S_u, V_\alpha)\right\} $.
	We first rewrite the inequality in Section~\ref{sec:discussion} as below:	\begin{align*}
		|R| \cdot \sum_{u \in V_\alpha} |S_u| \le & \sum_{u \in V_\alpha} 4 \min\{C\cdot w_v, |R|\} \cdot C \cdot w_v + \sum_{u \in V_\beta} 2 |R| \cdot  C \cdot w_v
	\end{align*}
To make this inequality holds, at least one term should be larger than $\frac{1}{2} |R| \cdot \sum_{u \in V_\alpha} |S_u|$, thus yielding the desired result.
\end{proof}

\medskip \noindent {\bf Generalized wHC Algorithm.} We extend the wHC algorithm for computing $R \times S$ with $|R| < |S|$ on a symmetric star topology.

 \begin{algorithm}
	\caption{{\sc BalancedPackingUnEqual}$(G, \dstr)$}
	\label{alg:balanced-packing-unequal}
	
	$L^* \gets L(R,S,V_C)$, $w \gets \max_v w_v$\; 
	\While{$\square$ is not fully covered}{
		$u \gets \arg \max_{v \in V_C} w_v$\;
		\uIf{$2^{-\ell} w  L^* \ge |R|$}{
			Assign to $u$ a rectangle of size $|R| \times (w_u \cdot L^*)$\;
		}
		\Else{
			$\ell \gets \arg \min_k \{ w \geq 2^k \cdot w_u \}$\;
			Assign to $u$ a square of size $(2^{-\ell} w  L^*) \times (2^{-\ell}w L^*)$\;
		}
		$V_C \gets V_C - \{u\}$\;	
	}
\end{algorithm}

To show the correctness of Algorithm~\ref{alg:balanced-packing-unequal}, it suffices to show that the grid is fully covered when $V_C$ becomes empty. 
Indeed, notice that each node $v$ covers an area of size at least ${L^* \cdot w_v} \cdot \min\{{L^* \cdot w_v},{|R|} \}$. Summing over all compute nodes, the area covered in total is at least
\[\sum_{v \in V_C}{L^* \cdot w_v} \cdot \min\{{L^* \cdot w_v},{|R|} \}  \ge|R| \cdot |S|\]
implied by~(\ref{eq:1}). Hence, the whole area of $\square$ is covered. 

Next, we analyze the cost of the algorithm. Observe that each node $v$ receives at most $4 L^* \cdot w_v$ tuples.
Hence, the cost is bounded by $4L^*$, yielding the following result.

\begin{lemma} \label{lem:whc-unequal}
	The wHC algorithm correctly computes the cartesian product $R \times S$ with (tuple) cost $O(C)$, where
	\[C = \max\left \{\max_v \frac{N_v}{w_v}, L(R,S,V_C) \right\}\]
\end{lemma}

\medskip \noindent {\bf Putting Everything Together on Symmetric Star.} Now we introduce our algorithm for computing cartesian product on a symmetric star. It can be easily checked that Algorithm~\ref{alg:cartesian-unequal} has its cost matching the lower bound in Theorem~\ref{thm:lb-cp-unequal-1} and Theorem~\ref{thm:lb-cp-unequal-2}, thus be optimal.
\begin{algorithm}
	\caption{{\sc GeneralizedStarCartesianProduct}$(G, \dstr)$}
	\label{alg:cartesian-unequal}
	
	\If{$\max_u N_u > N/2$}{
		all compute nodes send their data to $\arg\max_u N_u$\;
	}
	\Else{
		all compute nodes send their $R$-tuples to $V_\beta$\;
		Pick the best of:
		\begin{packed_enum}
			\item compute nodes send their data to $\arg \max_u w_u$\;
			\item all nodes in $V_\alpha$ send their $S$-tuples proportionally to $V_\beta$\;
			\item run wHC algorithm on $V_\alpha$ to compute $R \times \cup_{v \in V_\alpha} S_v$\;
			\end{packed_enum}
	}
\end{algorithm}

\end{document}